%% file: main.tex
\documentclass[11pt]{article}
\usepackage{graphicx} 
\usepackage{amsthm}
\usepackage{amsmath}
\usepackage{amsfonts}
\usepackage{algorithm}
\usepackage{algpseudocode}
\usepackage{listings}           
\usepackage{fullpage}
  \usepackage[multiple]{footmisc}
\usepackage{thmtools}
\usepackage{todonotes}
\usepackage{cleveref}
\usepackage{wrapfig}

\newtheorem{definition}[equation]{Definition}

\newtheorem*{theorem*}{Theorem}
\newcommand{\whp}{whp}

\newcommand{\myparagraph}[1]{\vspace{.05in}\noindent {\bf #1.}}
\newcommand{\rlfmis}{R-LFMIS}

\title{Parallel Batch-Dynamic Maximal Independent Set}
\date{}

\newtheorem{theorem}{Theorem}[section]
\newtheorem{corollary}{Corollary}[theorem]
\newtheorem{lemma}[theorem]{Lemma}

\newcommand{\LFMIS}[1]{M({#1})}

\newcommand{\InfluenceMIS}{I}
\newcommand{\enew}{E_{\text{new}}}
\newcommand{\gnew}{G_{\text{new}}}
\newcommand{\depth}{depth}

\author{Guy Blelloch\footnotemark[1] \\ guyb@cs.cmu.edu \and Andrew Brady\footnotemark[1]  \\ acbrady@andrew.cmu.edu \and  Laxman Dhulipala\footnotemark[2] \\ laxman@umd.edu \and Jeremy Fineman\footnotemark[3] \\ jf474@georgetown.edu \and   Jared Lo\footnotemark[4] \\ jaredlo@hawaii.edu}

\begin{document}

\maketitle

\renewcommand{\thefootnote}{\fnsymbol{footnote}}

\footnotetext[1]{Carnegie Mellon University, Pittsburgh, PA }
\footnotetext[2]{University of Maryland, College Park, MD}
\footnotetext[3]{Georgetown University, Washington, DC}
\footnotetext[4]{University of Hawaii, Honolulu, Hawaii}

\setcounter{footnote}{0}
\renewcommand{\thefootnote}{\arabic{footnote}}

\input{0_abstract}

\input{1_intro}

\input{2_prelims}

\input{3_influence}
\input{alt_algocode}

\input{5_algorithm}
\input{7_conclusion}

\myparagraph{Acknowledgements} This material is based upon work performed while attending the AlgoPARC Workshop on Parallel Algorithms and Data Structures at the University of Hawaii at Manoa, in part supported by the National Science Foundation under Grant CCF2452276. 
This work was also supported by NSF grants CCF2403235, CNS231719, CCF2119352,  CCF1919223, CCF1918989, and CCF2106759.

\bibliographystyle{plain}
\bibliography{bibliography/strings,bibliography/main}

\include{8_appendix}

\end{document}

%% file: 0_abstract.tex
\begin{abstract}

    We develop the first theoretically-efficient algorithm for maintaining the maximal independent set (MIS) of a graph in the parallel batch-dynamic setting. In this setting, a graph is updated with batches of edge insertions/deletions, and for each batch a parallel algorithm updates the maximal independent set to agree with the new graph.   A batch-dynamic algorithm is considered efficient if it is work efficient (i.e., does no more asymptotic work than applying the updates sequentially) and has polylogarithmic depth (parallel time).
    In the sequential setting, the best known dynamic algorithms for MIS, by Chechik and Zhang (CZ) [FOCS19] and Behnezhad et al. (BDHSS) [FOCS19], take $O(\log^4 n)$ time per update in expectation.
    For a batch of $b$ updates, our algorithm has $O(b \log^3 n)$ expected work and polylogarithmic depth with high probability (whp). 
    It therefore outperforms the best algorithm even in the sequential dynamic case ($b = 1)$.

    As with the sequential dynamic MIS algorithms of CZ and BDHSS, our solution maintains a lexicographically first MIS based on
    a random ordering of the vertices.    Their analysis relied on a result of Censor-Hillel, Haramaty and Karnin [PODC16]
    that bounded the ``influence set" for a single update, but surprisingly, the influence of a batch is not simply the union of the influence of each update therein.  We therefore develop a new approach to analyze the influence set for a batch of updates.
    Our construction of the batch influence set is natural and leads to an arguably simpler analysis than prior work. 
     We then instrument this construction to bound the work of our algorithm. To argue our depth is polylogarithmic, we prove that the number of subrounds our algorithm takes is the same as depth bounds on parallel static MIS.

    \end{abstract}

%% file: 1_intro.tex
\section{Introduction}

A maximal independent set (MIS) of an undirected graph is a maximal set of (pairwise) non-adjacent vertices. The MIS is one of the most fundamental graph structures, with many applications and extensive theoretical study \cite{Luby86,alon1986fast,BDHSS19,chechik19fully,CHK16,bateni2023optimal,BFS12,FN20,bernstein26adversary,ailon2008aggregating,daum2012leader,linial1987distributive,ghaffari2024near,nguyen2008constant}. 
In the sequential static setting,\footnote{By static, we mean that given a (single) graph, output a response; in contrast, the dynamic setting is one where the graph changes.} one can find an MIS on $n$ vertices and $m$ edges in $O(n+m)$ time in the following way. First, pick any permutation $\pi$ of the vertices. Then iterate through the vertices in $\pi$ order, skipping the vertices that are eliminated,
and otherwise adding the vertex to the MIS and eliminating it and its neighbors from the graph.
An MIS generated in this way is called a lexicographic-first MIS (LFMIS).

In the parallel static setting Karp and Wigderson~\cite{karp84fast} described deterministic and randomized polylog-depth parallel algorithms for finding an MIS, but their algorithms require at least $\Omega(n^2)$ work.  Luby~\cite{Luby86} improved the results with a pair of simple and influential randomized polylog-depth algorithms that require just linear work.\footnote{As Luby described the algorithms they use $O(m \log n)$ work, but they are easily modified to run in $O(m)$ work.  Alon, Babai, and Itai~\cite{alon1986fast} independently developed a very similar algorithm a few months later.}  All of these algorithms are based on finding an independent set, eliminating the neighborhood of this set, and repeating.  
Generating the LFMIS for a fixed order is known to be NC$^1$-complete for P~\cite{Cook85,Luby86}, and hence unlikely to have a polylogarithmic-depth solution. 
Blelloch, Fineman, and Shun~\cite{BFS12}, however, showed that on average, across
all orderings, finding the LFMIS is highly parallel; specifically, the longest chain of  dependencies, called the dependence depth, is polylogarithmic with high probability (whp)\footnote{Formally, we say that $f(n,c) = O(g(n))$ with high probability if there exists $c_0,k$ such that for $c \ge c_0$, there exists $n_0$ such that for $n \ge n_0$, $f(n,c) \le ck g(n)$ with probability at least $1-\frac{1}{n^c}$. }, and hence a natural parallelization takes polylog rounds whp.   
Fischer and Noever~\cite{FN20} then improved the bound to $O(\log n)$ depth whp.  These results imply another work-efficient polylog-depth randomized algorithm~\cite{BFS12}---pick a random $\pi$ ordering, 
and then in rounds add a vertex to the MIS as soon as it has no earlier uneliminated neighbor,
and eliminate it from the graph as soon as it has an earlier neighbor in the MIS.  Henceforth, we will use \rlfmis{} to indicate an LFMIS in random order.

In the sequential dynamic setting, Assadi, Onak, Schieber, and Solomon~\cite{assadi2018fully} present the first non-trivial
fully dynamic algorithm supporting updates in $O(m^{3/4})$ time (deterministic, worst case).  This was then
improved to $\tilde{O}(n^{2/3})$~\cite{du2018improved,gupta21simple}, and then to $\tilde{O}(m^{1/3})$ time, using randomization and assuming an oblivious
adversary~\cite{assadi2019fully}. The first sequential dynamic algorithms with polylog
time were developed independently by Chechik and Zhang (CZ)~\cite{chechik19fully} and by Behnezhad et al. (BDHSS)~\cite{BDHSS19}.  These algorithms
are randomized and both run in $O(\log^4 n)$ time assuming an oblivious adversary.\footnote{More precisely, BDHSS achieves a $O(\log^2 n \log^2 \Delta)$ bound, where $\Delta$ is the maximum degree of a vertex in the graph.  Das and Kuszmaul~\cite{das2026history} conjectured that by replacing the binary search tree data structure with a van Emde Boas tree, the runtime of BDHSS, at least for lexicographic-first maximal matching, can be improved to $O(\log^3 n \log \log n)$.} They are the current state
of the art.    Both these algorithms are based on the \rlfmis{} and rely on two crucial properties of a \rlfmis{}. The first, by Censor-Hillel, Haramaty and Karnin (CHK)~\cite{CHK16}, is that if an edge or vertex is inserted/deleted, the expected change in the MIS (the \emph{recourse}) is constant, and in fact at most 1.  
The second is that when processing the $i^{th}$ vertex in the $\pi$ order, the remaining
graph has degree $O(\frac{n \log n}{i})$ whp~\cite{BFS12}. BDHSS also relies on the property that the
\rlfmis{} dependence graph is shallow whp~\cite{BFS12,FN20}. 

A nice property of the CZ and BDHSS algorithms is that they are history independent---no matter the history of updates, the MIS depends only on the current graph and the permutation order, which is chosen at the start.
Note that the oblivious adversary assumption is necessary; an adaptive adversary would effectively know the random permutation and could create pathological cases accordingly. 

In this paper, we present results for MIS in the parallel batch-dynamic setting.   In this setting, one maintains a data structure supporting batches of updates that are processed in parallel.   One motivation of this setting is the desire to achieve good performance
on applications with high update rates, where processing updates sequentially could be too slow.
There has been extensive recent work on graph algorithms in this setting~\cite{simsiri2016work,AABD19,MDKLSW24,tseng2018batch,AABDW20, AB24,man26ufo,ABT20,TDS22,GT24, BB25,GK25sparsifiers,HM25,GK25coloring,liu2022parallel,GK25coreness,DLSY21}.  
Such parallel batch-dynamic algorithms have also been used to develop efficient static parallel graph algorithms~\cite{AB23,GGQ23,BGJV25}. We note there has been experimental interest in implementing theoretically efficient batch-dynamic graph algorithms \cite{IkramBAB25,MDKLSW24,man26ufo}. There has also been work in implementing batch-dynamic MIS specifically, though without theoretical guarantees \cite{trivedi24fast,nijhara25fast}. 
Typically, the goal of the batch-dynamic setting is for performance to scale linearly with the number of updates (i.e., if the best sequential dynamic algorithm has $f(n)$ work, the goal would be to achieve roughly $b f(n)$ work on a batch of size $b$), while also achieving polylog depth.  However, some results in this setting are actually able to achieve asymptotically less work per update by removing redundancies~\cite{IkramBAB25, man26ufo, AABDW20, tseng2018batch} in large batches. 
We know of no prior theoretically-efficient work in the batch-dynamic setting for the maximal independent set problem.
Our main result is the following.

\begin{theorem*}[Informal]
There exists an algorithm that maintains a maximal independent set under batch insertions and deletions of edges against an oblivious adversary, where for a size-$b$ batch of edge updates, the work is $O(b \log^3 n)$ in expectation and the depth is $O(\mathrm{polylog}(n))$ whp.
\end{theorem*}

We note that even when interpreted as a sequential algorithm, where there is only one update at a time, our result improves
the current state of the art~\cite{chechik19fully,BDHSS19}.
As with CZ and BDHSS, our results are based on \rlfmis{} and are history independent.  
To work in the batch setting, however, we need to develop a new approach to generalize the recourse bounds of CHK.  In particular,
the CHK bounds rely on the definition of an influence set for a single update---i.e., a set of vertices that could change from being
in or out of the MIS due to a single update. They show the size of the influence set to be at most one in expectation. 
Unfortunately, their definition of influence does not compose across a batch of updates---that is, contrary to intuition, the union of the influence sets for two or more updates made in any order does not necessarily cover the influence set for the updates made together (we show an example in the next section).   We therefore develop a significantly different technique and proof for bounding the size of a batch influence set. 

Towards this end, we define the influence analysis version of a problem. In this influence analysis, some input vertices and edges are initially set to undecided, meaning that their presence in the graph is unknown. As output, one wants to determine the full set of vertices whose MIS status is unknown (or influenced) as a consequence of the unknown inputs. 
We then show that for the influence analysis version of \rlfmis{}:
if $b$ input edges and vertices are undecided (can either be in or out of the graph), then the number of influenced outputs (vertices we are not sure
are in or out of the MIS) is $\leq b$ in expectation.   
This, in turn, implies the same bound on the number of vertices that can change
state (in and out of MIS) with a mix of $b$ edge insertions and deletions.
The approach we take for this proof is very different than in CHK due to the issue mentioned above.

We then develop an algorithm that has some similarities to CZ and BDHSS, but is designed for the parallel setting. 
In parallel, we process a subset of vertices $A$ whose status has changed (entered or left the MIS) to fix their neighbors as necessary. If we were to choose $A$ to be the  unprocessed influenced vertex with minimum permutation time, we would have a sequential algorithm. To process a vertex $v$ at permutation position $i$, we need to visit its two-hop neighborhood in the graph that remains at step $i$ of \rlfmis{}. 
We perform constant work for each of these visited vertices. 
We show that
the probability that \rlfmis{} step $i < n/4$ contributes a vertex to the influence set is at most $\frac{2bi}{n^2}$.
Along with
the $O(\frac{n \log n}{i})$ high probability bound on the degree, this leads
almost immediately to a bound of $O(\log^3 n)$ expected sequential time per update.

However, through a careful correctness argument and use of data structures, our algorithm also supports processing many vertices in parallel, even when this involves processing later permutation time vertices before earlier permutation time vertices. 
This comes with the drawback that we may require processing the same vertex multiple times, thereby increasing the work. 
Our parallel algorithm mitigates this issue, and achieves work-efficiency with our sequential variant, by processing vertices in a sequence of growing intervals of the permutation (which we call shells), similar to~\cite{BFS12}.
We also need to bound the number of parallel rounds (or dependence depth) of our algorithm, which does not immediately follow from the bound of the number of rounds of parallel \rlfmis{}.

\section{Technical Overview}

We begin by discussing the challenges with adapting the previous work (CHK~\cite{CHK16}, CZ~\cite{chechik19fully}, BDHSS~\cite{BDHSS19}) to the parallel batch-dynamic setting.\footnote{A reader interested in understanding our work without relation to prior work can skip to Section~\ref{sec:ourApproach}.}

\myparagraph{Challenge: Sequential influence set can iterate many rounds} 
The influence set (developed by CHK~\cite{CHK16}) is a superset of the vertices that actually change mark (actually enter or leave the \rlfmis{}). 
The influence set is defined iteratively. 
The base set $S_0$ is either empty if the \rlfmis{} is valid following the edge update $(u,v)$, or, if the \rlfmis{} is no longer valid, $S_0$ contains only the endpoint $v$ that is later in the permutation~$\pi$. The $r$th set $S_r$ includes all of $S_{r-1}$, plus the following vertices: the marked vertices that have at least one earlier neighbor in $S_{r-1}$ (these vertices may leave the MIS) and the unmarked vertices for which every marked neighbor is in the influence set (these vertices may join the MIS).
The influence set grows across iterations in this way until $S_r=S_{r+1}$ (until it stabilizes). We refer to the number of rounds of this propagation, i.e., the minimum $r$ for which $S_r = S_{r+1}$ as the \emph{influence propagation depth}. For an example of an influence set, see Figure~\ref{fig:exampleInfluence}.

At first glance, CHK's influence propagation process seems parallel, with parallel depth proportional to the number of rounds of influence set propagation, especially since the influence propagation depth is shown to be $O(1)$ in expectation. One would therefore hope to leverage known parallel \rlfmis{} upper bounds to get high probability bounds on the propagation depth. 
But even though parallel static \rlfmis{} finishes in $O(\log n)$ rounds whp, the depth of influence set propagation is not in fact always bounded by the depth of parallel \rlfmis{}: there exist graphs and permutations where parallel static LFMIS would take $O(\log n)$ rounds, but influence set propagation would take $\Theta(n)$ rounds.

Figure \ref{fig:sierra} illustrates such an example, where the influence propagation depth is very high but the dependence depth of LFMIS is low.  In particular, consider the simple path of vertices with permutation values $2$ to $n$ (for even $n$), where each vertex $i$ is connected to vertex $i+1$. (The figure illustrates $n=12$.) Also connect vertex 2 to vertex 4, and vertex 5 to every even vertex with permutation value at least $8$. In this graph, vertex $2$ and every odd vertex numbered at least 5 are marked (in the MIS). Now consider what happens when the edge $(1,2)$ is inserted. 
A static parallel LFMIS would take 3 rounds on the graph without the edge $(1,2)$, and it would take 4 rounds on the graph that includes that extra edge.  
But the influence set propagation depth here is $n-3$. 

The reason that the influence set propagation can take much longer than parallel \rlfmis{} on certain graphs and permutations is that influence set propagation must wait for all eliminators of a vertex to disappear before adding an initially unmarked vertex to the influence set, whereas dependence depth propagation is based on the earliest eliminator of a vertex. 
It is still theoretically possible that influence set propagation has a $O(\text{polylog}(n))$ whp bound, but this does not follow from prior work, and we do not attempt to show this. We instead opt to propagate in a different manner designed to inherit the round upper bound from the dependency depth of the graph.

\begin{figure}
\includegraphics[width=.4\textwidth]{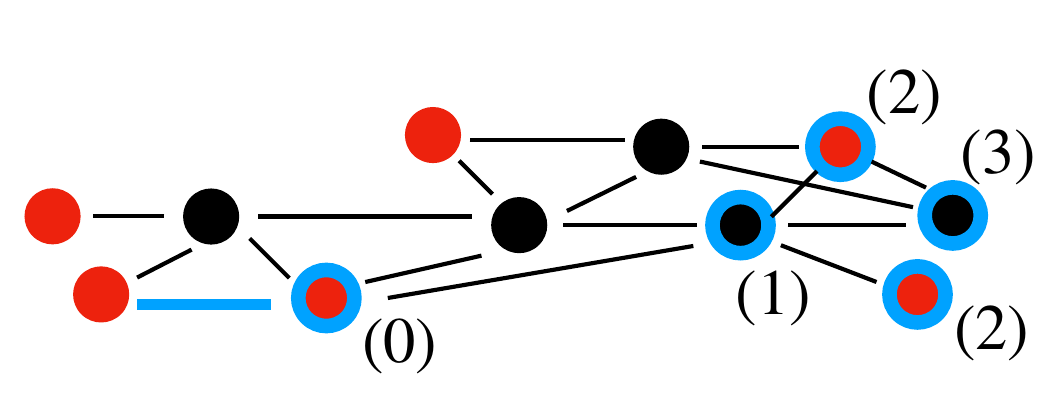}
\hspace{.3in}
\includegraphics[trim = {.8cm 1.5cm .8cm 1.2cm}, clip, width=.5\textwidth]{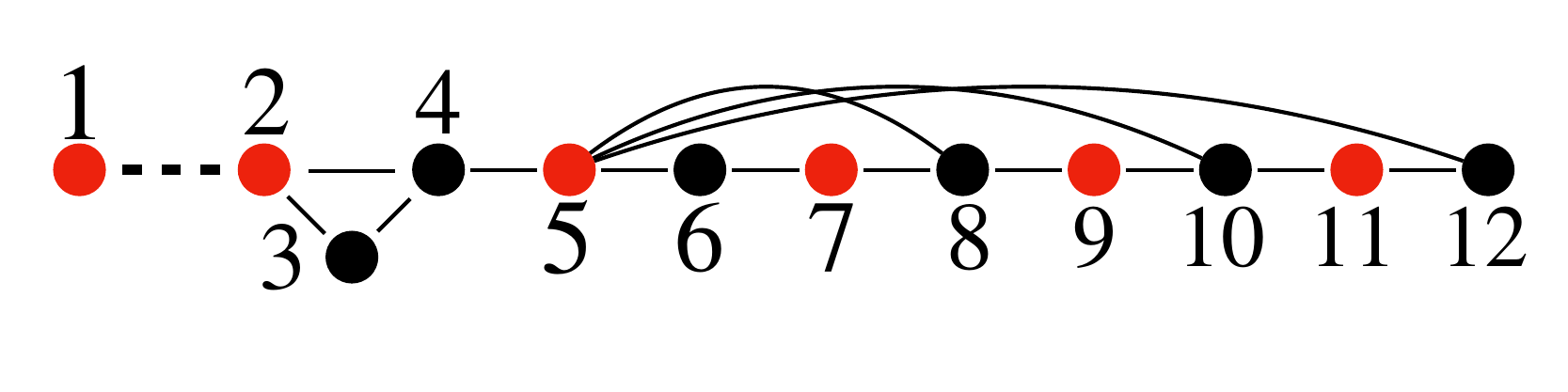}

\caption{Example of (sequential) influence set from CHK~\cite{CHK16}. The MIS of the graph (before the edge update) is marked in red. Here, the blue edge is the edge being inserted. Vertices circled in blue are influenced. In parenthesis is the round when they join the influence set. We have drawn vertices in $\pi$ order from left to right (the leftmost vertex is the earliest in the ordering).  \label{fig:exampleInfluence} }
 \caption{Example graph where influence round propagation time is greater than dependency depth. The vertices are drawn from $\pi$-value left to right (the red vertices thus have permutation values 1,2,5,7,9,11). The red vertices are the marked vertices before the edge update. Consider inserting the dashed edge $(1,2)$ to the graph. The resulting MIS would be $\{1,3,5,7,9,11\}$. The influence set would be $\{2,3,\ldots,12\}$. Specifically, $S_0=\{2\}$,$S_1=\{2,3,4\}$, $S_2=\{2,3,4,5\}$,$S_3=\{2,3,\ldots,6\}$, \ldots, $S_9 = \{2,3,\ldots,12\}$. Thus, influence set propagation takes $9$ rounds. Before the edge update, parallel static \rlfmis{} would take 3 rounds, and after the edge update, parallel static \rlfmis{} would take 4 rounds. By extending the length of the path (as described in main text), the influence set propagation can take $n-3$ rounds, whereas the parallel static \rlfmis{} would continue to take 4 rounds. \label{fig:sierra}}
\end{figure}

\myparagraph{Challenge: Constructing a batch influence set} When trying to construct a batch influence set, a natural proposal is to simply change the base case. That is, start with $S_0$ equal to the set of all conflicts immediately caused by the batch insert, but keep the same definitions for $S_r$ for $r \ge 1$. Note that for $b=1$, this definition is equivalent to CHK's definition. Surprisingly, however, the batch influence set on a batch $B$ of edge updates can be greater than the union of the influence sets of applying these updates one at a time sequentially to the graph. Thus, CHK's proof does not bound the size of this batch influence set. 

For an example of this issue, see Figure \ref{fig:batch-bloat}. Note that the influence set is a superset of the vertices that actually change mark: the influence set can make ``mistakes," which then cause additional vertices to mistakenly join the influence set. When performing updates one at a time, the mistakes from the influence set of one update are erased before the next update begins. However, in a batch update, these mistakes can compound and cause even more vertices to join the influence set. Figure \ref{fig:batch-bloat} is a small example of such a compounded mistake: the triangles cause the individual mistakes, which combine to influence the starred vertex. 

Moreover, their proof method critically relies on there being only a single edge update. In their proof, on an edge update $(u,v)$ where $u < v$, they consider the influence set $S'$ that would occur if $v$ were moved to the front of the permutation. Then, they argue that either $S=S'$ or $S$ is empty, and that with (informally) $\frac{1}{|S|}$ probability, $S$ is empty. This all-or-nothing approach fails in the batch setting, because some edges could succeed to cause influence, other edges could fail, and the edges that succeed could combine forces to cause influence neither could cause individually.

\begin{figure}
\includegraphics[trim={.5cm .2cm .5cm .5cm},clip,width=\linewidth]{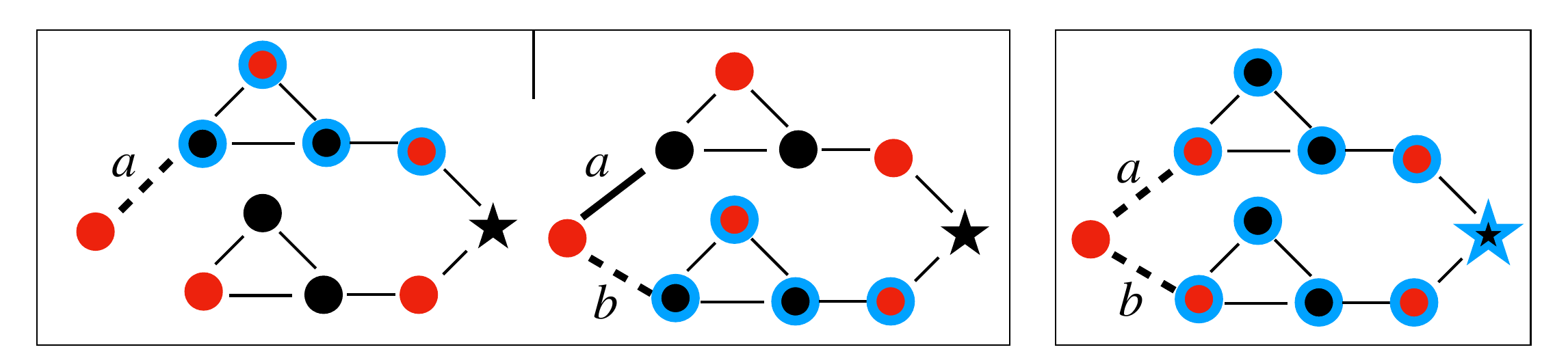}
\caption{\label{fig:batch-bloat}Consider the above graph, where vertices that are marked (in the old graph, before the batch update) are red. Consider inserting edges $a$ and $b$. In the left rectangle, the edges are inserted one at a time, first $a$, and then $b$. The influence set for each insert is highlighted in blue. Note that the star vertex is not influenced when $a$ and $b$ are inserted separately, but is influenced when $a$ and $b$ are inserted together. }
\end{figure}

\myparagraph{Challenge: Bounding the runtime} The runtime arguments of CZ and BDHSS do not transfer to the batch setting for technical reasons. We include a discussion in Section \ref{sec:transferHard}. 

\subsection{Our approach \label{sec:ourApproach}}

Our approach involves (1) defining an influence-analysis version of \rlfmis{}, (2) proving that for $b$ unknown
inputs this version has a batch influence set of expected size at most $b$, (3) relating this directly to the vertices that can possibly change mark in
a set of $b$ edge insertions and deletions, (4) developing a generic algorithm that operates either sequentially
or in parallel for propagating changes through the graph, (5) bounding the work of this propagation, and (6) bounding the depth.
We outline these steps here.

(1) We first define an ``influence analysis'' version of \rlfmis{} where one is given a graph along with an arbitrary subset of its vertices and edges marked as ``undecided'', indicating that we do not know if they are present or not.
It outputs a set of ``influenced'' vertices for which we consequently might not know if they are in the MIS or not.
The set of influenced vertices can be identified using a modified version of the greedy \rlfmis{} algorithm. 
As in LFMIS, the algorithm iterates through the vertices in $\pi$ order, skipping vertices
that are eliminated.
When processing a vertex $v$ that is known to be in the graph (not in the undecided set, and hence in the MIS),
the algorithm eliminates any neighbors of $v$ across decided edges (as with \rlfmis{}), and 
marks any neighbors of $v$ across undecided edges as undecided (we do not know if they are eliminated).
When the algorithm processes a vertex $v$ that is undecided, it does not know if $v$ exists 
and hence whether it is in the MIS. Therefore, it adds $v$ to the influence set
and marks all its neighbors as undecided (since, again, we do not know if they are eliminated).
The vertex $v$ and incident edges are eliminated (as with \rlfmis). See Figure \ref{fig:ourInfluence} for an example of how a graph changes in response to a step of influence analysis.  We only use this algorithm for analysis, we never run it.

(2) To analyze the size of the influenced set, it is relatively straightforward to show that
any expansion of uncertainty to neighbors is offset in expectation by the consumption of undecided edges or the elimination of other undecided edges and vertices. 
We then use this conservation to
show that for the $b$ undecided edges and vertices at the start, the expected total size of the influence set $I$
at the end is $\leq b$ (roughly, each step never grows the expected sum of $b +|I|$ and $b = 0$ at the end).
We believe that even in the single update case, this influence analysis proof is arguably simpler and more intuitive than the original proof.

(3) Similar to CHK, we note that for a vertex to begin marked and end unmarked, an earlier neighbor (or new neighbor from insertion) must end marked, and for a vertex to begin unmarked and end marked, all earlier marked neighbors must have unmarked (or have been disconnected via deletion). If the influenced set is a superset of all vertices before $v$ that change mark, $v$ will also join the influence set if necessary; therefore by induction the influence set includes the earlier neighbors that change mark. 

\begin{figure}
\centering
\includegraphics[trim = {0cm 4cm 0cm 0cm}, clip, width=.5\linewidth]{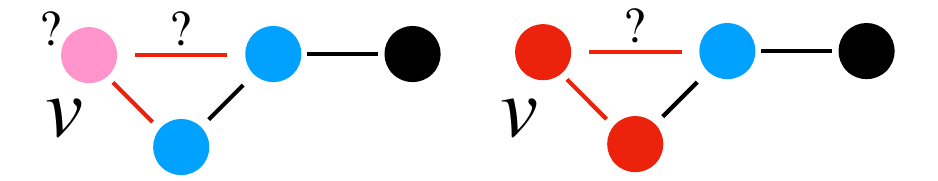}
    \caption{Our batch influence set propagation. The question marks indicate undecided vertices and edges before processing the next vertex~$v$. Both images show different cases for processing vertex $v$ depending on whether it is undecided (left) or not (right). The colors indicate the states after processing $v$: pink means influenced, red means eliminated, and blue means undecided. 
The left shows the case where $v$ is an undecided vertex; when processing $v$, $v$ joins the influence set (pink), its two incident edges are removed (red), and both its neighbors become undecided (blue). The right shows the case where $v$ is not undecided. Here its bottom neighbor is eliminated because it is connected by a not-undecided edge, but its right neighbor becomes undecided because the connecting edge is undecided. 
\label{fig:ourInfluence}}
\end{figure}

(4) Although the batch influence set is a great analysis tool, using it as our algorithm would be too slow. For our actual algorithm, we assign to each vertex an elimination time, which is the permutation time of its earliest marked neighbor, or its own permutation time if no earlier neighbor is marked. Our generic parallel algorithm then operates by repeating the following: find all conflicts in the graph (wherever there is a vertex with an incorrect elimination time), choose a subset of the vertices causing these conflicts, and fix the conflicts caused by that subset. We get a work-depth tradeoff based on how we choose the subset.

The correctness of this algorithm follows from the fact that we keep iterating until all conflicts are removed, that a set of marked vertices without conflicts is necessarily the LFMIS, and that our algorithm makes progress on every step and therefore terminates. 

(5) At a high level, the efficiency argument proceeds as follows. First, we show that the work of the algorithm is bounded by the remaining 2-hop neighborhood of each vertex in the influence set (at the permutation time of that vertex). This corresponds to the neighbors (and neighbors of neighbors) we access looking for a replacement edge when a vertex's eliminator leaves, as well as the cost of updating our out-edge structure. 
Then, using our batch influence set, we will show that the sum of the 2-hop neighborhood sizes is $O(b \log^2 n \log \Delta)$. 

In order to bound our work by the 2-hop neighborhood size of influence set vertices, we process influenced vertices in $O(\log^3 n)$ shells (permutation time ranges). By fully processing earlier vertices before later vertices, we ensure that we do not access many out-neighbors of any vertex in later shells. We use $O(\log^3 n)$ shells so that each vertex has a low probability of sharing a live neighbor in its shell. 
Note that influence set vertices can be repeatedly marked and unmarked (see Figure \ref{fig:repeat_process}), and that a given vertex can be processed a superconstant number of times. 
However, if we process by shells, and a vertex has no neighbors in its shell, we will only process that vertex once. 
Thus, we will only redundantly process a few vertices, fixing our work bounds.

\begin{figure}
\centering
    \includegraphics[trim={0cm 2.5cm 0cm 0cm}, clip, width=.7\linewidth]{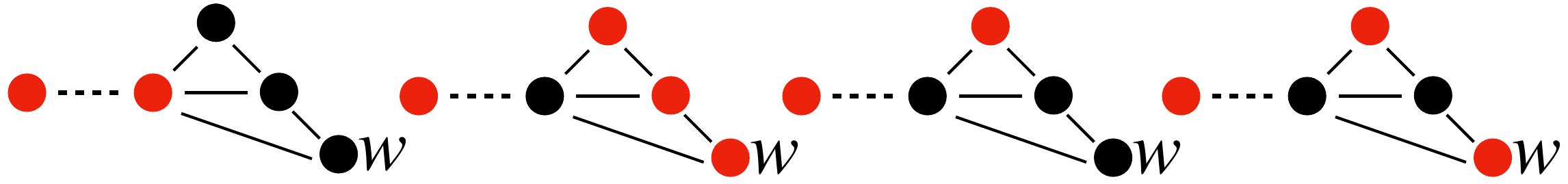}
    \caption{Example of repeating processing of a vertex. In response to the edge insert (dotted), we show selected snapshots from the graph as propagation occurs. Note that vertex $w$ will be repeatedly marked and unmarked as we update its neighbors. In some graphs, we could process the same vertex $O(\log n)$ times. To avoid this when possible, we separate vertices into shells, and process different shells separately. }
    \label{fig:repeat_process}
\end{figure}

To show that the expected sum of the size of the 2-hop neighborhoods is $O(b \log^2 n \log \Delta)$, we first prove that a step $i < n/4$ of \rlfmis{} has roughly a $\frac{bi}{n^2}$ chance of contributing a vertex to the influence set, which will then pay $O((\frac{n \log n}{i})^2)$ cost. Then we have expected $O(\frac{b \log^2 n}{i})$ contribution per step, which leads to the desired bound. Note that steps where $i \ge n/4$ are inexpensive because the degree of the surviving graph is $O(\log n)$ whp. 

The proof that each step $i < n/4$ has a low probability of contributing to the influence set informally proceeds as follows. 
As before, the expansion of uncertainty to neighbors is offset by the maybe vertices being eliminated. 
Because there is a (roughly) $\frac{b}{n}$ chance a batch edge will be selected, causing a vertex to join the undecided vertex set, the expected size of the undecided vertex set is roughly upper bounded by $\frac{bi}{n}$. Since an (alive) undecided vertex has a $\frac{1}{n-i}$ chance of being chosen in step $i$, we get the probability $\frac{bi}{n^2}$. 

(6) For our depth bound, we argue that the number of rounds we need to process a shell is bounded by the dependence depth of the graph, which is $O(\log n)$ whp. Our overall depth then becomes $O(\log^5 n)$ by (pessimistically) multiplying all of the following: we have $O(\log^3 n)$ shells (empty shells are skipped, but for the depth bound we apply the pessimal bound); each shell takes $O(\log n)$ rounds to process; each round uses parallel primitives whose depth is $O(\log n)$ even for large batches.

%% file: 2_prelims.tex
\section{Preliminaries}

\myparagraph{Parallel model} We assume the binary fork-join model \cite{blelloch2019optimal}. In this model, a process can spawn two new processes with a fork, which can operate independently, and later join together when finished. 
The model permits concurrent writes; in particular, if two threads try to write $a$ and $b$ concurrently to the same variable, an arbitrary choice of $a$ and $b$ will be written.
Nested parallelism is permitted: these new processes can themselves spawn more processes. 
The work of a computation is the total number of operations required across all processors, and the depth is the longest chain of dependent instructions.

We permit the atomic operations test and set (TS) and compare and swap (CAS), which are only used in our parallel primitives. 
We note that the work is robust across models: permitting randomization, the work of an algorithm in the binary fork-join model is the same asymptotically as the work in the arbitrary way fork-join model and as the work in the CRCW PRAM, and the depth of an algorithm is within a $O(\log n)$ factor of the depths in the arbitrary way fork-join and CRCW PRAM \cite{blelloch2019optimal}. We thus care about optimizing the work: optimizing the depth is not a focus in this work.

To simplify our pseudocode, the pseudocode includes an atomic write minimum operation. This can be easily implemented under the hood with a semisort, for more details see Section \ref{sec:atomics}.

\myparagraph{Parallel primitives} We will use various standard parallel operations. Given an array, a scan (also known as prefix sum), will return a running sum of the associative operator $f$ on the first $i$ elements of the array. A scan can be used to implement a filter, which given a boolean function and an array, returns an array of the elements for which the function returns true. A scan can also be used to give a flatten, which given an array of arrays, returns a single array of all of the elements. Scan, filter, and flatten all take $O(n)$ work and $O(\log n)$ depth \cite{jaja1992parallel}.

Given an array of elements, a semisort will return an array such that elements with the same key are consecutive (though elements with different keys not sorted from least to greatest). Semisort requires $O(n)$ expected work and $O(\log n)$ depth whp \cite{blelloch2019optimal,gu2015top}. 

A bag data structure supports appending elements, deleting elements (given a pointer), returning the number of elements present, and outputting $b$ (arbitrary elements) of a container. A bag does not support membership queries and is unsorted. Recent work developed a bag that supports a batch of $b$ element updates in $O(b)$ work and $O(\log b)$ depth \cite{blelloch26faster}, and that can output $b$ elements in $O(b)$ work and $O(\log b)$ depth.

\myparagraph{Batch-Dynamic model} In the batch-dynamic model, edge updates arrive in discrete (integer) time steps. At time $t$, a batch $B_t$ of edge insertions and deletions arrive and are applied to the old graph $G_{t-1}$, yielding the new graph $G_t$. In response to this batch of edges, we must update the MIS. At time $t$, we must have an MIS for the graph at $G_t$: each vertex must know whether it is marked or unmarked, and we must hold a bag containing all marked vertices. 

The edge updates are being given by an oblivious adversary, which knows our algorithm but not our randomness, and must commit to an entire update sequence at the beginning of time. The update sequence is revealed to the algorithm by time step (at time $t$, the algorithm does not know what the edge updates for time $t+1$ will be). We note that the oblivious adversary is a standard assumption for dynamic MIS \cite{chechik19fully,BDHSS19}.

\myparagraph{MIS and LFMIS}
    Formally, the set $M$ is a maximal independent set (MIS) iff the following conditions hold: 

    \begin{enumerate}
        \item If $u \in M$, then $N(u) \cap M=\emptyset$ (no adjacent marked neighbors), where $N(u)$ is the neighbors of $u$. 
        \item $u \not\in M$, that there exists $v \in M \cap N(u)$. 
        
    \end{enumerate}
    We call vertices in the MIS marked, and vertices not in the MIS unmarked. 

A natural static algorithm for finding an MIS is the lexicographically first, or greedy, approach. The vertices are given an arbitrary order $\pi$. We call $\pi(v)$ the {\em permutation time} (or $\pi$-value) of $v$. We iterate through this order. If a vertex has been eliminated, we continue iterating, and do not add the vertex to the MIS. Otherwise, we add the vertex to the MIS, then eliminate it and its neighbors. 
This is called the lexicographically first MIS (LFMIS). When $\pi$ is chosen uniformly at random, the resulting MIS is called the \rlfmis{}. 
Given a set $M$, one can verify that it is the LFMIS without iterating, by verifying that every marked vertex has no marked neighbors and every unmarked vertex has an earlier marked neighbor, for details see Section \ref{sec:lfmisVerify}.

\myparagraph{Dependency Graph}
Traditionally in an LFMIS, let the deletion time be the permutation time at which a vertex marks or is eliminated \cite{BFS12}. Then, the dependence graph is an orientation of the graph by the deletion time (low to high). In the event of tied deletion time, a marked vertex points to an unmarked one. If both are unmarked, then the vertex with lower permutation time points to the higher. Note that a vertex $v$'s out neighbors are those present when $v$ is eliminated from the graph. In this way, a dependency graph nicely captures how a graph shrinks over time.
Note that the dependency graph is {\it not} an orientation by $\pi$-ordering. If a vertex with late $\pi$ value is eliminated early, it will have many out neighbors, as shown in Figure \ref{fig:dependency}. 

\begin{figure}

\centering
    \includegraphics[trim={.5cm 2cm 0cm 0cm}, clip,width=.5\linewidth]{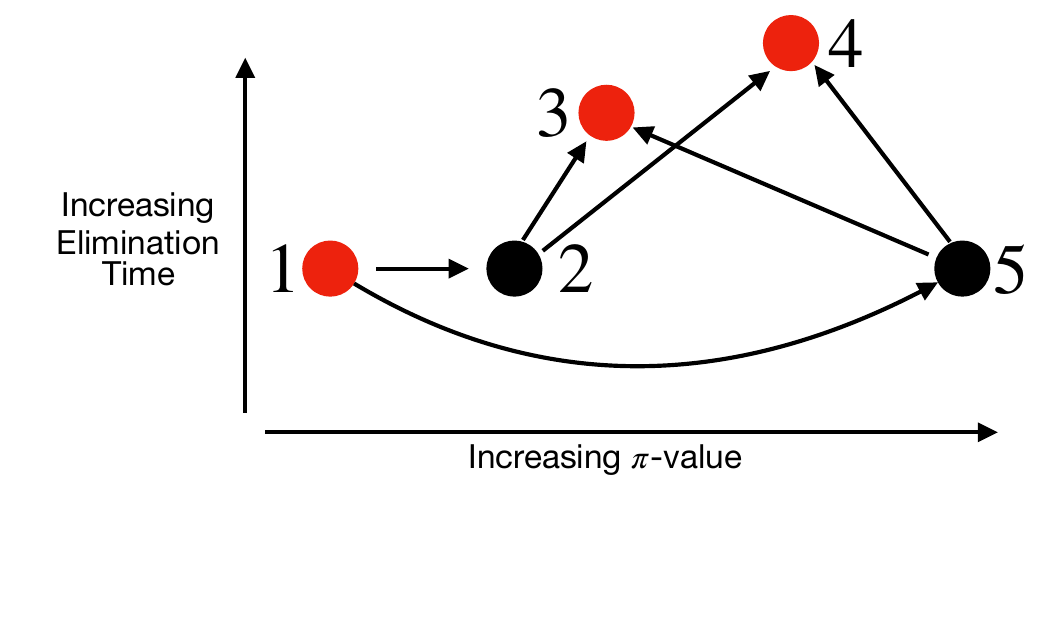}

    \caption{Example graph, marked vertices are shaded red. In this example, vertex 5 has later $\pi$-value than vertex 3 but earlier deletion time, so we orient the edge from 5 to 3, not the other way around. \label{fig:dependency}}
\end{figure}

\myparagraph{Parallel static MIS} Note that the LFMIS can be computed in parallel, where each vertex that has the local highest priority of its surviving neighbors marks and eliminates its neighbors. 
 Blelloch, Fineman, and Shun ~\cite{BFS12} showed that when iterating the \rlfmis{} and removing eliminated neighbors, the graph shrinks in degree quickly, as expressed by the following lemma.\footnote{By iterating, we mean stepping through the vertices and adding them to the LFMIS if possible. Concretely, suppose we have processed the first $i-1$ vertices, and removed all the neighbors of vertices we marked. BFS showed that the maximum degree in this remaining graph is small whp.} 

\begin{lemma}
    Let $\pi$ be a random order of the vertices, and suppose we have iterated through the first $i-1$ vertices (marking if alive, and eliminating neighbors of marked vertices). Then the remaining graph has max degree $O(\frac{n \log n}{i})$ whp \cite{BFS12}. \label{lem:degreeBound}
\end{lemma}

The number of rounds required by parallel \rlfmis{} is equal to the length of the longest alternating eliminator path (AE-path).  An {\it alternating eliminator (AE)-path} is a marked-unmarked alternating path of increasing $\pi$ where the marked-unmarked edges are eliminator-target pairs.\footnote{
An AE-path was called a dependency path in Fischer-Noever. Note that we differ from Fischer-Noever in allowing the path to start with an unmarked or a marked vertex; note that this alters the path length by at most 1, and so their results still apply within constant factors. We prefer the name AE-path to clarify that we are not trying to bound the length of any path in the dependency graph, but only paths that cause a bottleneck for \rlfmis{}. In particular, unmarked vertices can have harmless directed edges between them, and these will not appear in an AE-path} Blelloch, Fineman, and Shun showed that the longest AE-path has $O(\log^2 n)$ length whp \cite{BFS12}, which was improved to $O(\log n)$ whp by Fischer and Noever \cite{FN20}.

%% file: 3_influence.tex
\section{Influence Analysis of MIS \label{sec:influence}}

\begin{table}
\centering
    \begin{tabular}{c|c}
    Symbol & Meaning \\
    \hline
    $V^?$ & undecided vertices \\
    $E^?$ & undecided edges \\
    $V_i^?$ & undecided vertices, before iteration $i$ \\
    $E_i^?$ & undecided edges, before iteration $i$ \\
    $\pi$ & chosen permutation on the vertices \\
    $V_i^l$ & vertices not in $\pi[1:i)$ \\
    $F$ & shorthand for \mbox{InfluenceMIS} \\
    $M$ & the \rlfmis{}, before the edge updates \\
    $V_i$ & alive vertices before iteration $i$ \\
    $M(G)$ & The LFMIS for $G$ (with order $\pi$) \\
    \end{tabular}
    \caption{Notation for influence analysis.}
\end{table}

The influence analysis version of a function determines which outputs
are undecided given that some inputs are undecided.  More formally, for
deterministic problems (i.e. functions) taking sets to sets, we define it as follows.

\begin{definition} 
Consider a deterministic problem $P : \mathcal{P}(\mathcal{I}) \rightarrow \mathcal{P}(\mathcal{O})$.
The \emph{influence analysis} problem for $P$ is
given an input $I \subseteq \mathcal{I}$ and an undecided subset $I^? \subseteq I$, to determine
the undecided (influenced) outputs
$$O_? = \{o \in \mathcal{O} \mid \exists_{I^x \subseteq I^?} \mbox{ s.t. } o \in (P(I \setminus I^x) \bigtriangleup P(I))\}.$$
\end{definition}

Looking forward to its application to dynamic algorithms, note that influence analysis makes no distinction between insertion and deletion of elements: the set $I$ includes all input (i.e. the input with the inserts added but without the deletes removed), and $I^?$ will include both inserted and deleted elements. 

A \emph{conservative solution} to an influence analysis problem is one that    
reports back all the undecided output elements, but might include some extra
elements.
Note that returning all elements is a vacuously conservative solution; we want to find a conservative solution with small size. 
We note that influence analysis is loosely related to the sensitivity analysis of boolean functions and stability analysis \cite{acar04dynamizing,odonnell14analysis}.  
                                                                        
In our case, the source problem $P$ that we care about is LFMIS.  The input of the influence
analysis problem for LFMIS is a     
graph $G = (V,E)$ and an undecided set of vertices $V^? \subseteq V$ and edges $E^? \subseteq E$.  If a vertex is undecided, it means that it together with its incident edges in $G$ are either present or not.  The problem
is to determine which outputs (members of the MIS) are influenced by the
undecided vertices and edges.
Figure~\ref{algo:influence} presents a variant of greedy MIS that conservatively solves the          
influence analysis problem for LFMIS. Throughout the paper, we will use $G_i = (V_i,E_i)$, $V_{i}^?$, and $E_{i}^?$ to
indicate the values of $G = (V,E)$, $V^?$, and $E^?$ at the beginning of step~$i$. Note that counting the iteration number $(i)$ will become helpful in proving the efficiency of our parallel algorithm in Section \ref{sec:efficiency}, but is not needed for the results in this section: from the view of this section, we could have defined influence propagation in purely recursive terms.

\begin{figure}
\begin{lstlisting}[numbers=left,firstnumber=1,xleftmargin=2em,basicstyle=\small\sffamily, keywordstyle=\bfseries, mathescape=true, escapeinside={@}{@}, morekeywords={parfor,for,if,else,foreach,return,while}, columns=flexible]
InfluenceMIS($G = (V,E)$, $V^? \subseteq V$, $E^? \subseteq E$, $i$)
  if $V^? = \emptyset$ and $E^? = \emptyset$ : return $\emptyset$
  $u = \pi(i)$
  if $u \notin V$ : 
    return InfluenceMIS$(G, V^?, E^?, i+1)$
  if $u \in V^?$
    foreach $v \in N_G(u)$ : add $v$ to $V^?$
    delete $u$ from  $G, V^?, E^?$
    return $\{u\} \cup \mbox{InfluenceMIS}(G, V^?, E^?, i+1)$
  else
    foreach $v \in N_G(u)$ :
      if $(u,v) \in E^?$ : add $v$ to $V^?$
      else : delete $v$ from $G, V^?, E^?$
    delete $u$ from $G, V^?, E^?$
    return InfluenceMIS$(G, V^?, E^?, i + 1)$
\end{lstlisting}
\caption{Conservative algorithm for the influence analysis problem for LFMIS. The algorithm returns a (possible overestimate) of all vertices that might be influenced by the undecided inputs. $V_1$ is the original set of vertices in some ordering $\pi$. Eliminating a vertex from a graph removes it and its incident edges. Removing $u$ from $E^?$ removes any edges from $E^?$ that include $u$.}
\label{algo:influence}
\end{figure}

\begin{lemma} If $u \in V^?$, then $u \in V$ (all undecided vertices are alive). 
\end{lemma}

\begin{proof}
    By construction. Each place in our code where we remove a vertex $v$ from $V^?$, we also remove $v$ from the graph. Therefore, any undecided vertex is still alive.  
\end{proof}

\begin{lemma}
The algorithm InfluenceMIS conservatively solves the influence analysis problem for LFMIS based on the ordering
given by $\pi$.
\label{lemma:batchinfluenceworks}
\end{lemma}
\begin{proof}

The proof is by induction on graph size.
The inductive hypothesis is that the lemma holds on the smaller graph, assuming that $V \subseteq \{V_{\pi}(j) : j \in \{i...n\}\}$, i.e., we will visit all vertices.
We note that all steps maintain $V \subseteq \{V_{\pi}(j) : j \in \{i...n\}\}$ since advancing $i$ always removes $V_{\pi}(i)$
from the graph.
For the base case where nothing is undecided, it correctly returns the empty set.    For the inductive case, we consider the subcases.
We can ignore steps where $u \notin V_i$ since they do not do anything beyond advancing $i$.   If $u$ is
undecided, then it may or may not appear in the MIS depending on whether it is present or not, and we therefore
include it in the influenced outputs. Furthermore, since we do not know whether $u$ is there, we also do not know
whether it will eliminate its neighbors in \rlfmis{}. Therefore, the neighbors become undecided (note that
some could already be undecided).   No other vertices become undecided, so when making the recursive call and
assuming its result is correct by induction, it will return all influenced vertices.   If $u$ is not undecided,
then it will always be added to the MIS and therefore not returned as influenced. Furthermore, since $u$ is
added to the MIS, any neighbors through decided edges will be eliminated, and any neighbors through undecided
edges might or might not be eliminated.   Again, assuming that the recursive call is correct by induction, the
final result will include all influenced vertices.   
\end{proof}
We note that the algorithm might return vertices that will always be in or out for all possible settings of
$V^?$ and $E^?$. In particular, if a vertex $v$ gets added to $V^?$ twice (or more), then it is possible
that the two vertices that add $v$ always have the opposite parity (one is in if the other is out), and hence $v$ will
never be in the MIS (one such example is vertex $5$ in Figure \ref{fig:sierra}).   Furthermore, if $v$ adds an adjacent $u$ to $V^?$, the vertex $u$ could be returned
as influenced even though $u$ will always be in the MIS.   In the following, however, we show
that the influence set is small in expectation. We note that the influence set of CHK also contained vertices that would be definitely in or out of the \rlfmis{}. 

We add some notation. Let $\Pi(V_1)$ be the set of all possible permutations of all vertices. Let $\Pi_i(V_1)$ be the set of all possible permutations, given that the vertices in permutation positions $[1, i)$ are fixed to some value. Let $V_i^l$ denote the vertices for which permutation values have not yet been assigned, before a vertex has been chosen for permutation position $i$. Let $V_i^?$ denote the undecided vertices after iteration $i$, and $E_i^?$ the undecided edges after iteration $i$. We use $F$ as a shorthand for $\mbox{InfluenceMIS}$.

\begin{lemma} \label{lem:expected_change}
    Let $G,V_1^?,E_1^?$ be the starting graph and undecided vertices and edges.
    Let $G,V^?,E^?$ be a possible state of the graph and undecided sets immediately before iteration $i$. Let $\pi[1:i)$ be a possible prefix of a permutation that would lead to this graph state. 
    Consider $F(G,V^?,E^?,i)$, and let $F(G',V'^?,E'^?,i+1)$ be the resulting recursive call. 
    Let $\Pi_i$ be the set of permutations that could result in $G,V^?,E^?$ from the given starting state, given that $\pi[1:i)$ is fixed. 
    Draw $\pi$ uniformly at random from $\Pi_i$.
    Then $E_{\pi \sim \Pi_i}[|V'^?|+|E'^?|] \le |V^?|+|E^?|-\frac{|V^?|}{n-i+1}$. 
\end{lemma}

\begin{proof}
    
Consider the choice of $u=\pi(i)$. Note that for any permutation $X \in \Pi_i$, any ordering of the vertices in $V_i^l$ would have been in $\Pi$. Therefore, the choice of $u$ is uniformly random among the unchosen vertices. Note that if $u \not\in V$, then we iterate $i$ without changing the graph.

Instead of taking a vertex-centric perspective, viewing the neighbors that are added or removed from the undecided set due to picking a particular vertex, we take an edge-centric view on the change. 
We break down the change in $|V^?|+|E^?|$ to three factors: that a chosen vertex in $V^?$ removes itself, the removal of undecided edges when an endpoint is picked, and that an edge can remove/add one of its endpoints. We thus have that 

\begin{align*}
E[|V'^?|+|E'^?|-|V^?|-|E^?|]
  &= -\sum_{u \in V^?} \frac{1}{n-i+1}
     - \sum_{e \in E^?} \frac{2}{n-i+1} \\
  &\quad + \sum_{\substack{(u,v) \in E^? \\ u,v \not\in V^?}} \frac{2}{n-i+1} + \sum_{\substack{(u,v) \in E^? \\ u \in V^?, v \not\in V^?}} \frac{1}{n-i+1}  \\
  &\le \frac{-|V^?|}{n-i+1}
\end{align*}

We explain this equation in more detail. 
First, the expected effect on $|V'^?|+|E'^?|$ of a vertex removing itself is $\frac{-|V^?|}{n-i+1}$. Similarly, the expected change in the size of $E^?$ is $\frac{-2 |E^?|}{n-i+1}$. For two vertices $u,v$ not in the undecided set where $(u,v) \in E^?$, there is a $\frac{2}{n-i+1}$ chance of selecting one and adding the other to the undecided set. For two vertices $u,v$ where $u \in V^?$ but $v \not\in V^?$ and $(u,v) \in E^?$, there is a $\frac{1}{n-i+1}$ chance of picking $u$, causing $v$ to be added to $V'^?$. 
Note that a decided edge with only one endpoint in $V^?$ has a $\frac{1}{n-i+1}$ chance of adding one vertex to $V^?$, and a $\frac{1}{n-i+1}$ chance of removing one vertex from $V^?$, for a net zero change in $|V^?|+|E^?|$.

Thus the overall expected change in $|V^?|+|E^?|$ (compared to $|V^?|+|E^?|$) is upper bounded by $\frac{-|V^?|}{n-i+1}$. Note that when there are undecided edges between undecided vertices, the expected change in $|V^?|+|E^?|$ is at most $\frac{-|V^?|}{n-i+1}$.
\end{proof}

In the above lemma, we have shown that the total number of undecided vertices and edges shrinks over time in expectation. In the next lemma, we will show that the growth of the influence set is  balanced by the shrinking of the undecided set, and so the size of the influence set is nicely bounded.

\begin{lemma}
\label{lem:influenceBound}
For $\pi$ drawn uniformly at random from all permutations of $V$, we have that:
\[ E_{\pi \sim \Pi(V)}[|\mbox{InfluenceMIS}(G, V^?, E^?,i)|] \leq |V^?| + |E^?| \]
\end{lemma}
\begin{proof} We proceed with induction.

\noindent\textbf{Base Case 1:} Suppose that for some $i$, we have that $|V| \le |V^?| + |E^?|$. By construction, the amount that can still be added to the influence set is no more than $|V|$, so $F(G,V^?,E^?,i) \le |V^?| + |E^?|$.

\noindent\textbf{Base Case 2:} Suppose that for some $i$, $V^?$ and $E^?$ are empty. Then no more is added to the influence set, so $F(G,V^?,E^?,i) \le 0 = |V^?| + |E^?|$.

\noindent\textbf{Inductive Hypothesis:}
Let $j \ge i+1$.
Let $\pi[1:j)$ be any prefix of a permutation that would yield $G,V^?,E^?$ after $j-1$ iterations.
Let $\Pi_j$ be the set of possible permutations of $V$ with $\pi[1:j)$ fixed, and suppose that $\pi \sim \Pi_j$.
Then we have that $E_{\pi \sim \Pi_j}[F(G,V^?,E^?,j)] \le |V^?| + |E^?|$.

\noindent\textbf{Inductive Step:}
Now consider $E_{\pi \sim \Pi}[F(G,V^?,E^?,i)]$.

Let $G',V'^?,E'^?$ denote the graph and undecided sets resulting from choosing a vertex $u$.
We condition the expectation on the choice of the next vertex $u$.
Observe that
\begin{align*}
E_{\pi \sim \Pi_i}[F(G,V^?,E^?,i)]
  &= \sum_{u \in V_i^l} \frac{1}{n-i+1} E_{\pi \sim \Pi_i | \pi(i) = u}[F(G',V'^?,E'^?,i+1)] + \frac{|V^?|}{n-i+1}.
\end{align*}
Because we have fixed a prefix of the permutation of length $i$, we can apply the inductive hypothesis to get that $E_{\pi \sim \Pi_{i+1}}[F(G',V'^?,E'^?,i+1)] \le |V'^?| + |E'^?|$. Plugging back into our main equation yields
\begin{align*}
E_{\pi \sim \Pi_i}[F(G,V^?,E^?,i)]
  &= \frac{|V^?|}{n-i+1} + \sum_{u \in V_i^l} \frac{1}{n-i+1} (|V'^?| + |E'^?|) \\
  &= \frac{|V^?|}{n-i+1} + E_{u \in V_i^l}[|V'^?| + |E'^?|]
\end{align*}

which by Lemma~\ref{lem:expected_change} is upper bounded by $|V^?| + |E^?| -\frac{|V^?|}{n-i+1} + \frac{|V^?|}{n-i+1} = |V^?| + |E^?|$ as desired.
\end{proof}

\subsection{Connecting influence analysis to \rlfmis{}}

We now want to relate the influence analysis problem for LFMIS to the
dynamic algorithm for the problem.  In this paper, we only care about adding and removing edges. 
As before, we define $V_i$ and $V_i^?$ to be the states of $V$ and $V^?$ at the beginning of the $i^{th}$ call to \mbox{InfluenceMIS}. Let $M(G)$ be the LFMIS on a graph $G$ for a permutation $\pi$ that is clear from context. Let $\enew = E \cup E_+ \setminus E_-$, $\gnew = (V,\enew)$, $E_{all}=E \cup E_+$, $G_{all}=(V,E_{all})$, and $I=
\mbox{InfluenceMIS}((V,E_{all}), \emptyset, E_+ \cup E_-,1)$

\begin{lemma} \label{lem:remove_everywhere}
    If vertex $u=\pi(i) \in V_j^?$ for some $j\le i$, then either $u \in V_i^?$ or $u \notin V_i$.
\end{lemma}

\begin{proof}
    The claim is true if no predecessors modify $u$'s placement after iteration $j$. Suppose $w=\pi(k), j<k<i$ is a predecessor of $u$. The only way for $w$ to remove $u$ from $V^?$ is to also remove $u$ from $V$.
\end{proof}

\begin{lemma} \label{lem:lfmisg_markings} 
    Let $G=(V,E)$ be a graph, and suppose we run \mbox{InfluenceMIS} on an input graph $G_{all}=(V, E \cup E_+)$ and $E^?$ so that $E^?$ contains the batch updates. If vertex $u=\pi(i)$ and $u \in V_i\setminus V_i^?$, then $u \in M(G)$. If $u=\pi(i)$ and $u \notin V_i$, then $u \notin M(G)$.
   
\end{lemma}

\begin{proof}
    By induction on the $\pi$ value of $u$. The base case is that $u=\pi(1)$. Vertex $u$ will always be in $V_1$ and marked regardless of the arrangement of edges. 
    
    Assume now that it is true for all vertices in $\pi[1,i-1]$. Let $\pi(u)=i$. Suppose $u \in V_i \setminus V_i^?$, and consider a predecessor $w=\pi(j)$ of $u$ in $G$. 
    
    Suppose for contradiction that $w \in V_j \setminus V_j^?$. Then $w$ either removes $u$ from $V$ (if $(w,u) \not\in E_j^?$) or adds $u$ to $V_j^?$ (if $(w,u) \in E_j^?$). The first case contradicts that $u \in V_i$. In the second case, by Lemma~\ref{lem:remove_everywhere}, we have that either $u \in V_i^?$ or $u \notin V_i$, contradicting $u \in V_i \setminus V_i^?$. Therefore, $w \not\in V_j \setminus V_j^?$. 

    Suppose for contradiction that $w \in V_j^?$. Then $w$ adds $u$ to $V^?$, and by Lemma~\ref{lem:remove_everywhere}, we have that $u \in V_i^?$, contradicting $u \notin V_i^?$. Therefore, $w \notin V_j^?$.

    Thus $u$ only has predecessors (in $G$) of the form $w=\pi(j)$ such that $w \notin V_j$. By the inductive hypothesis, $w \notin M(G)$. Therefore, $u$ is not eliminated by any predecessor in $G$, and is in $M(G)$.

    Suppose now that $u \notin V_i$. Then $u$ is eliminated from $V$ by some predecessor $w=\pi(j) \in V_j \setminus V_j^?$. Furthermore, $(w,u)$ must exist in $G$ because $(w,u)$ is not an undecided edge.

    By the I.H., $w \in M(G)$, thus $u \notin M(G)$.
\end{proof}

\begin{lemma} \label{lem:in_intersection}
    If $u=\pi(i) \in M(G)$ and $u \in V_j^?$ for some $j<i$, then $u \in V_i^?$. 
\end{lemma}
\begin{proof}
    The contrapositive of Lemma~\ref{lem:lfmisg_markings} indicates that $u \in M(G) \implies u \in V_i$. The vertex $u$ cannot be removed from $V^?$ during some iteration $k>j$ without also permanently removing $u$ from $V$, so it must be the case that $u$ remains in $V^?$ until the beginning of iteration $i$.
\end{proof}

\begin{lemma}
Consider LFMIS applied to graphs $G = (V,E)$ and $G_{new}=(V,\enew)$. We have
\[M(V,E) \bigtriangleup M(V, \enew) \subseteq \mbox{InfluenceMIS} ((V, E_{all}), \emptyset, E_+ \cup E_-,1) \]
\end{lemma}

{

\begin{proof}
    By induction on the $\pi$ values of each vertex. We claim that for vertex $u$ with $\pi(u)=i$, if $u \in M(G) \bigtriangleup M(G_{new})$,  then $u \in 
    \InfluenceMIS$.

    Base case: For $\pi(u) = 1$, $u$ will always be marked regardless, thus $u$ is not in the left-hand side. 

    Inductive case: Assume that the claim holds for all vertices $\in \pi[1,i-1]$. We want to show that the claim holds for $\pi(u)=i$.
    
    \begin{enumerate}
        \item[Case 1:] $u=\pi(i)$ is a vertex in $M(G) \setminus M(\gnew)$. Either a predecessor $w$ of $u$ in $G_{new}$ flips to marked (and $(w,u)$ exists in both the old and new graph), or a new predecessor $w$ of $u$ is connected to $u$ via insertion so that $w$ is marked in $G_{new}$ and eliminates $u$.
        \begin{enumerate}
            \item[Case 1a:] Some predecessor $w=\pi(j)$ of $u$ in $G$ flips to marked. By the I.H.,
            $w \in \InfluenceMIS$, which moreover implies that $w \in V_j^?$.
            It must be the case, then, that $w$ adds $u$ to $V^?$.
            By Lemma~\ref{lem:in_intersection}, $u \in V_i^?$, and thus adds itself to InfluenceMIS during iteration $i$.

            \item[Case 1b:] No predecessor of $u$ in $G$ that flips to marked, and instead a new predecessor $w=\pi(j)$ of $u$ is connected via insertion so that $w$ is marked in the new graph and eliminates $u$. If $w \in M(G)$, then $w \in V_j$ by Lemma~\ref{lem:lfmisg_markings}. Since $(w,u) \in E_j^?$, $u$ is added to $V_j^?$.

            On the other hand, if $w \notin M(G)$, then $w \in \LFMIS{\gnew} \setminus \LFMIS{G}$, and by the inductive hypothesis, $w \in \InfluenceMIS$, i.e., $w \in V_j^?$. Because we run \texttt{InfluenceMIS} on $G_{all}$ consisting of all edges in either $G$ or $G_{new}$, $w$ adds its new neighbor $u$ to $V^?$. 
            
            In either case, because $u$ is added to $V^?$ by some vertex with an earlier permutation time than $u$, by Lemma~\ref{lem:in_intersection}, $u \in V_i^?$, which results in $u$ adding itself to InfluenceMIS during iteration $i$. 
        \end{enumerate}

        \item[Case 2:] $u=\pi(i)$ is a vertex in $\LFMIS{G_{new}} \setminus \LFMIS{G}$. Because $u \notin \LFMIS{G}$, by Lemma~\ref{lem:lfmisg_markings}, $u \notin V_i \setminus V_i^?$. Therefore, to guarantee that $u \in \InfluenceMIS$, we need only show that $u \in V_i$. We prove this is true by contradiction. Suppose $u \notin V_i$. Then there must be a predecessor $\pi(j)=w \in V_j \setminus V_j^?$ that eliminates $u$ from $V$ during the influence analysis. By Lemma~\ref{lem:lfmisg_markings}, $w \in \LFMIS{G}$. If $(w,u)$ is an undecided edge, then $w$ adds $u$ to $V^?$ and fails to remove $u$ from $V$. Otherwise, because $u \in \LFMIS{G_{new}}$, we have $w \notin \LFMIS{G_{new}}$, so $w \in \LFMIS{G} \setminus \LFMIS{G_{new}}$. By the I.H., $w \in \InfluenceMIS$, which can only occur if $w \in V_j^?$, which contradicts that $w \in V_j \setminus V_j^?$. Because both cases lead to contradiction, $u \in V_i$, hence $u \in \InfluenceMIS$.
        
    \end{enumerate}
\end{proof}
}

\begin{corollary}
For a graph $G$, the expected recourse for a set of $b$ edge updates (mix of insertions and deletions) is at most $b$.
\end{corollary}

\begin{proof}
Note that the recourse in batch-dynamic MIS is the symmetric difference of the MIS before and after the edge update. 
    By the above lemma, the symmetric difference of the \rlfmis{} before and after the edge updates is bounded by the size of the influence set, which by Lemma \ref{lem:influenceBound} is bounded by $b$.
\end{proof}

%% file: alt_algocode.tex
\begin{figure}
\begin{lstlisting}[basicstyle=\small\sffamily, keywordstyle=\bfseries, mathescape=true, escapeinside={@}{@}, morekeywords={parfor,for,if,else,foreach,return,while}, columns=flexible]
In the following pseudocode, the vertex names (e.g. $u,v,w$) indicate their ordering in the permutation.
The variable $l(u)$ maintains the deletion time of vertex $u$: $\text{if}$ $l(u) =_\pi u$ then $u$ is in the MIS
The variable $O(u)$ is true $\text{if}$ $u$ used to be in MIS (one round before)
Assume edges are directed, where $u \rightarrow v$ $\text{if}$ $(l(u),u) <_\pi (l(v),v)$.
$N^h(u)$: out edges of form $u \rightarrow v$ where $u <_\pi v$
$N^l(u)$ : out edges of form $u \rightarrow v$ where $u >_\pi v$
$N^+(u) = N^h(u) \cup N^l(u)$
$N^-(u)$ : the in edges stored as $\log n$ bags, the $i^{th}$ bag stores $\pi$-values in range $[2^i,2^{i+1})$
$N^-(u,v)$ : fetch the in edges $w$ of $u$ such that $w \geq_\pi v$ 
\end{lstlisting}

\noindent\begin{minipage}{.46\textwidth}
\begin{lstlisting}[numbers=left,firstnumber=1,xleftmargin=2em,basicstyle=\small\sffamily, keywordstyle=\bfseries, mathescape=true, escapeinside={@}{@}, morekeywords={synchronize,parfor,for,if,else,foreach,return,while}, columns=flexible]
lower $v$ to $u$:
  $N \leftarrow N^-(v,u)$
  writeMin$(l(v), u)$
  if $(l(v) = u)$ :
      parfor $w \in N$ : orient$(v,w)$
      if $O(v)$ : enque$(Q, (v,-))$

raise $v$ from $u$:
  $N \leftarrow N^l(v) \cup N^h(v)$
  $t(v) \leftarrow \mbox{findReplacement}(v)$
  $l(v) \leftarrow t(v)$
  parfor $w \in N$ :  orient$(v,w)$
  if $l(v) =_\pi v$ : enque$(Q, (v,+))$

findReplacement$(v)$ :
  $R \leftarrow \{w \in N^l(v) \mid l(w) =_\pi w\}$
  if $R = \emptyset$ : return $v$ // no replacement
  else return $\min(R)$

orient$(v,w)$ :
  if $(l(v),v) <_\pi (l(w),w)$: $v \rightarrow w$
  else: $w \rightarrow v$
\end{lstlisting}
\end{minipage}
\hspace{1em}
\begin{minipage}{.46\textwidth}
  \begin{lstlisting}[numbers=left,firstnumber=23,xleftmargin=2em,basicstyle=\small\sffamily, keywordstyle=\bfseries, mathescape=true, escapeinside={@}{@}, morekeywords={parfor,for,if,else,foreach,return,while}, columns=flexible]
propagate($E_+$, $E_-$) :
  empty $Q$
  $L \gets \{ \{u,v\} \in E_+,u<_\pi v \mid O(u) \wedge l(v) >_\pi u \}$
  parfor $\{u,v\} \in L$: lower $v$ to $u$

  $L \gets \{ \{u,v\} \in E_-, u <_\pi v \mid O(u) \land l(v) =_\pi u \}$
  parfor $\{u,v\} \in L$: raise $v$ from $u$

  while $Q \neq \emptyset$ :
    $A \subseteq Q$, $Q = Q \setminus A$
    parfor $(u,+) \in A, $ :  // lowering stage
      if $l(u) =_\pi u$ :
        $O(u) \leftarrow$ true
        parfor $v \in N^h(u)$ : lower $v$ to $u$ 

    parfor $(u,-) \in A$ :  // raising stage
      if $l(u) \neq_\pi u$ :
        $O(u) \leftarrow$ false
        $L \gets \{N^h(u) \mid l(v)=_\pi u \}$
        parfor $v \in L$ :  raise $v$ from $u$

\end{lstlisting}
\end{minipage}
\caption{Our Parallel Change Propagation Algorithm. Given a set of edges $E_+$ to insert and a set of edges $E_-$ to delete, the algorithm updates the \rlfmis{}. The runtime efficiency of the algorihtm depends on how $A$ is picked. If $A$ is the single minimum $\pi$-value vertex in the queue, then we get a sequential algorithm. If we pick $A=Q$, we get a maximally parallel algorithm, with lower depth but slightly higher work. If pick all up to some cutoff (a single shell), then we get our main algorithm.
 \label{alg:main}}

\end{figure}

%% file: 5_algorithm.tex
\section{Algorithm \label{sec:algorithm}}

\begin{table}
\centering
    \begin{tabular}{c|c}
    Symbol & Meaning \\
    \hline
    $u$ & the vertex $u$, also its permutation time $\pi(u)$ \\
    $l(u)$ & elimination time of vertex $u$ \\
    $<_\pi$ & comparator on permutation time \\
    $Q$ & conflict queue \\
    $A$ & subset of conflict queue \\
   
    \end{tabular}
    \caption{Notation used in Section \ref{sec:algorithm}. }
\end{table}

\subsection{Data Structures}
We will store the graph oriented by elimination time. In this section, we will write $u$ to denote the vertex $u$, as well as its permutation time $\pi(u)$. We use the comparator $<_\pi$ to emphasize that we are comparing permutation times: $u <_\pi v$ means that $\pi(u) < \pi(v)$. We write $l(u)$ for the elimination time of a vertex. The graph will be oriented from least to greatest, according to the value of the pair $(l(u),u)$, as shown in Figure \ref{fig:dependency}. When we compare $(l(u),u)$ to $(l(v),v)$ we first compare $l(u)$ and $l(v)$, then tiebreak on $u$ and $v$ if $l(u)=l(v)$. 

We will keep track of the following quantities per vertex. Each vertex has three types of neighbors: \emph{in neighbors}, \emph{out higher neighbors}, and \emph{out lower neighbors}. For an edge $(u,v)$, $u$ is an in neighbor of $v$ if $(l(u),u) < (l(v),v)$. The vertex $u$ is an out (higher) neighbor of $v$ if $(l(u),u) > (l(v),v)$ and $u > v$. The vertex $u$ is an out (lower) neighbor of $v$ if $ (l(u),u) > (l(v),v)$ and $u < v$. For a vertex $v$, we denote the in neighbors by $N^-(v)$, the out higher neighbors by $N^h(v)$, and the out lower neighbors by $N^l(v)$. 

In neighbors are sorted into $\log n$ bags of doubling range (of elimination times), and unsorted within each bag. The $i^{th}$ bag contains the in neighbors with elimination times in range $[2^i,2^{i+1})$. Out lower, and out higher, neighbors are unsorted.

A vertex holds the in-edge bags in an array of length $O(\log n)$. The first time that we insert an edge adjacent to endpoint $v$, we initialize this array, which takes $O(\log n)$ work. 

\subsection{Algorithm}

We give pseudocode for our algorithm in Figure \ref{alg:main}. Recall that our program is operating in synchronized lock step (i.e. joins after each line).

\myparagraph{Helper functions}
Given a set of pairs $(j,i)$, the lower $j$ to $i$ routine (lines 1-6) will do the following (to all pairs, in parallel). The new elimination time (level) for $j$ will be set to the minimum of the $i$'s that are lowering it (line 3). Then, some of $j$'s edges will be flipped to correspond to the new elimination time (line 5). If $j$ was previously marked, then we add it to the queue, so that it can propagate changes to its neighbors (line 6).

Given a set of vertices $(j)$, the raise $j$ from $u$ routine (lines 8-13) will do the following (in parallel lock step). The vertex $j$ will search for a replacement eliminator (earlier marked vertex) (line 10). If one is found, the minimum such vertex (line 18) will be assigned as $j$'s new elimination time (line 11). Otherwise, $j$ will mark (line 11), and join the queue to propagate changes to its neighbors (line 13). Some of $j$'s edges are flipped to account for its new elimination time (line 12).

\myparagraph{Main Algorithm}
Given an insert edge batch $E_+$ and a deletion edge batch $E_-$, we do the following. For the insertion batch, we find all adjacent marked neighbors, and lower the later neighbor (lines 25-26). For the deletion batch, we find all update edges where the earlier end is marked and eliminates the later unmarked endpoint, and we raise the later endpoint (lines 28-29).

Our queue\footnote{Technically, the vertices that we are processing are stored in an array of bags, one per shell, not a queue. We refer to this structure as a queue for intuition, because in the sequential case it would be a priority queue.} may now contain vertices that must trigger changes in their neighbors. While the queue is nonempty, we take a subset $A$ of the queue to process. 
We remove $A$ from the queue. Then, for the vertices which previously marked, we lower their neighbors (lines 33-36), and for the vertices that previously unmarked, we raise their eliminated neighbors (lines 38-42).

Regardless of the subset we pick, our algorithm is correct (see Section \ref{sec:correctness}). However, the subset we choose affects our work and depth bounds. Roughly speaking, the more vertices we choose, the more parallel our algorithm is, but we may do redundant work as a result. Note that we get a sequential algorithm at a high-level similar to BDHSS when we choose a single vertex, the minimum permutation time vertex, to process each round.

\subsection{Correctness \label{sec:correctness}}

In this subsection, we will show the correctness of our algorithm. First, we will show that we correctly maintain the oriented graph on the edges. Assuming that our orientation is correct, we show that our queue will contain all of the information necessary to fix our \rlfmis{}. We then show that we make progress in every iteration, and so our algorithm is correct. 

\myparagraph{Maintaining a correct orientation}
Recall that 
    we say that a graph orientation is correct according to $\pi$ if for all edges $(u,v)$, we have that $u \rightarrow v \iff (l(u),u) <_\pi (l(v),v)$.
In this subsection, we will show that we maintain a correct graph orientation, between each lowering and raising stage. 
Whenever we call the orient function (lines 20-22), we correctly fix the edge called on. 
For efficiency reasons, we do not call orient on every edge incident on a vertex when we change its elimination time, only edges that may actually flip direction. We show that we call orient on the necessary edges to fix the orientation.

\begin{lemma}
    Consider a set of calls lower$(v_i,u_i)$, where $u_i < v_i$. These calls maintain a correct graph orientation. 
\end{lemma}
\begin{proof}

Within the call to \texttt{lower}, each $v : (u,v)$ (where $u <_\pi v$) has its elimination time potentially modified. Of the edges $(u_i,v)$, the smallest $u_i$ will win the writeMin and set the elimination time of $v_i$ to $u_i$.
    Denote a vertex's original elimination time by $l(x)$ and new elimination time by $l'(x)$.
    
    We must show that all neighbors of $v$ with new elimination time in the range $[u_i,l(v)]$ have an orient call. 

    First, note that lower $v$ to $u_i$ will call orient on all edges with neighbor originally in range $[u_i,l(v)]$. 

    Now consider a neighbor $w$ of $v$ with $l'(w) \in [u_i,l(v)]$ but $l(w) \not\in [u_i,v]$. Since lower only reduces elimination time, we have $l(w) > l(v)$, and so $v$ was an in-neighbor of $w$. Then $w$ had its own lower call to a value $l'(w) \in [u_i,v]$. Since $w$'s lower call will check all in-neighbors with elimination time originally in range $[w',w]$ and $l(v) \in [w',w]$, $w$'s lower call will orient $(v,w)$. Therefore, all in neighbors of $v$ in range $[u_i,l(v)]$ have an orient call. 
\end{proof}

\begin{lemma}
    Consider a set of calls $raise(v_i,u_i)$, where $u_i < v_i$. These calls maintain a correct graph orientation. 
\end{lemma}

\begin{proof}
    Note that raise can only increase the elimination time of $v_i$, and thus it suffices to only \texttt{orient} every (formerly) outgoing edge $(v,w) \in N^+(v)$, which raise does by orienting every edge to a neighbor in $N$ (line 9).
\end{proof}


\begin{lemma}
At the beginning of each lowering/raising stage, the graph is oriented correctly.  
\end{lemma}

\begin{proof}

 Note that whenever we call \texttt{orient} on an edge, we correctly direct the edge based on the elimination time. Furthermore, note that we change elimination times, then sync, then call orient (no interleaving of calls to orient and changes of elimination time).

  If we called orient on every edge where at least one vertex changed elimination time, then the result would follow immediately. However, for efficiency reasons, calling orient on all such edges is wasteful. What we need to show is that if we change the elimination time of $u$ from $t$ to $t'$, that we call orient on all neighbors of $u$ with elimination time in range $t$ to $t'$. 

  After the step of initializing $L$ and directly before using \textbf{parfor} to lower vertices in parallel, no vertices have had their elimination times nor edge orientations modified yet, so the dependency DAG is correct. By the above lemmas, each lower and raise maintains the graph orientation. Because all elimination time changes are expressed with a lower or raise, the graph orientation is correct between each raising and lowering stage by induction. 
\end{proof}

\myparagraph{Maintaining the conflict queue}
Suppose that we have assigned elimination times to every vertex (where $l(u) \le u$ for all $u$). Recall that $u$ is marked iff $l(u)=u$.

There are two types of conflicts:
\begin{enumerate}
\item Too early conflicts: $l(u) < u = l(x)$ ($x$'s elimination time is too early).
\item Too late conflicts: $l(u) = u < l(x)$ ($x$'s elimination time is too late).
\end{enumerate}

Note that if there are no conflicts, then we have found the \rlfmis{}. 
Thus, conflicts are pairs $(u,x)$ (which are necessarily also edges in the graph). Note that $x$ could have multiple too late conflicts, but only one too early conflict. Note that $x$ could have a too early and too late conflict at the same time (which is fine because we do not raise and lower concurrently).
In this subsection, we will show that our queue captures all conflicts currently present in our graph. Note that in our queue, we do not store $(u,x)$, but only $(u,+)$ (if $u$ had marked) or $(u,-)$ (if $u$ had unmarked); we will show that this is all we need to store.

\begin{lemma}
    Suppose that $(u,x)$ is a conflict where $l(u) < u = l(x)$. Then calling \texttt{raise} $x$ from $u$ will resolve this conflict. 

    Similarly, suppose that $(u,x)$ is a conflict where $l(u) = u < l(x)$. Then calling \texttt{lower} $x$ to $u$ will resolve this conflict. \label{lem:adjustFixes}
\end{lemma}
\begin{proof}
    Raising $x$ will increase its elimination time, so that $u < l(x)$, fixing the conflict. Note that only one call to raise $x$ can happen concurrently, so this will fix. 

    Lowering $x$ will reduce $l(x)$ at least to $u$, if not lower (because of concurrent lower calls). Thus, $l(x) \le u$, so the conflict is fixed after the lower call.
\end{proof}

Once we have marked $u$, if we need to propagate changes to $u$'s neighbors, we want to add $(u,+)$ to the queue. Similarly, if we unmark $u$ and need to propagate changes, we want to add $(u,-)$ to the queue. In this next lemma, we show that we actually add these values to the queue when necessary.

\begin{lemma}
    Consider the conflicts created by a lower $v$ to $u$ call (the pairs $(v,w)$ where $l(v) < v = l(w)$). If any such conflicts occur, then $(v,-)$ is added to the queue. 

    Similarly, consider the conflicts created by a raise $v$ from $u$ call. If any such conflicts occur, then $(v,+)$ is added to the queue. 

    Therefore, any conflicts created by a lower or raise call correspond to an addition to the queue. \label{lem:conflictsAdded}
\end{lemma}
\begin{proof}
    Consider the lower $v$ to $u$ call. Reducing $v$'s elimination time could cause later neighbors of $v$ to have an incorrect (too early) elimination time. For this to occur, $v$ would need to be marked before the lower call, in which case $O(v)$ would be true, and so $v$ would be enqueued. Note that any neighbor $w$ of $v$ which now has an incorrect elimination time has a later permutation time than $v$, so we have $l(v) < v = l(w)$. When there are concurrent lower calls, each of which would enqueue $v$, the winner (minimum $u$) will enqueue $v$.

    Consider a raise $v$ from $u$ call. Increasing $v$'s elimination time could cause a later neighbor $w$ of $v$ to have an incorrect (too late) elimination time. For this to occur, $v$ would need to mark (raise to $l(v)=v$), which would cause $v$ to enqueue. Note that $l(v)=v < l(w)$. 
\end{proof}

\begin{lemma}
    Processing $(u,+)$ fixes all too-late conflicts involving $u$, and processing $(u,-)$ fixes all too-early conflicts involving $u$. \label{lem:conflictFixed}
\end{lemma}
\begin{proof}

  We will want to show that processing $(u,+)$ fixes all too late conflicts, and processing $(u,-)$ fixes all too early conflicts.

  First, consider $(u,+) \in A$, and let $(u,x)$ be a too-late conflict, where $l(u)=u < l(x)$. Note that $u$ passes the $l(u)=u$ check, and will call lower on all neighbors in $N^h(u)$, and since $u < x$ we have that $x \in N^h(u)$, so we do call lower $x$ to $u$. By Lemma \ref{lem:adjustFixes}, this does fix the conflict.

  Next, consider $(u,-) \in A$, and let $(u,x)$ be a too-early conflict, where $l(u) < u = l(x)$. Since $u=l(x)$ and $u < x$, we have that in the raising round, $x \in N^h(u)$ and so we will raise $x$ from $u$. By Lemma \ref{lem:adjustFixes}, this does fix the conflict.
\end{proof}

\begin{lemma}
    At the beginning of the while loop and between each mark and unmark stage, for all too-early conflicts $(u,x)$ we have $(u,-) \in Q$ and for all too-late conflicts $(u,x)$ we have $(u,+) \in Q$. 
\end{lemma}

\begin{proof}

    At the beginning of propagate, the queue is empty, and only later batch endpoints can have conflicts. Running lower on the batch insertions will resolve all too-late conflicts involving inserted edges via Lemma \ref{lem:conflictFixed}.  By Lemma \ref{lem:conflictsAdded}, lowering (some of) the batch endpoints will correctly add the subset of these batch endpoints that cause conflicts to the queue.

    At this point in time, since edge deletion cannot create too-late conflicts, note that there are no more too-late conflicts in the graph. There are two (possibly overlapping) types of too early conflicts: those incident on an edge deletion (where the earlier endpoint was the old eliminator), and those adjacent to a batch insertion endpoint (where the batch insertion endpoint was the old eliminator). The latter type of conflict is reflected already in the queue. 

    Let $(u,v)$ be an edge deletion with a too early conflict. Necessarily $u$ was marked. Then this edge $(u,v)$ will join $L$, and so we will call raise on this vertex, fixing the conflict (by Lemma \ref{lem:conflictFixed} and adding $(v,+)$ to the queue if more conflicts are created (by Lemma \ref{lem:conflictsAdded}).

    Note that raising vertices cannot create too-early conflicts. All too-early conflicts caused by edge deletions were processed, and all of the other too-early conflicts were already present in the queue. Therefore, all too-early conflicts in the graph are present in the queue.
    
    Note that the only too-late conflicts that can be (currently) present are those caused by raising the later endpoints of deleted edges, which were added to the queue by Lemma \ref{lem:conflictsAdded}. Therefore, all too-late conflicts created were added to the queue, and so the queue is up to date with all conflicts.

    Now we will proceed inductively. 
    Consider a vertex $(u,+) \in A$ participating in the lowering. By Lemma \ref{lem:conflictFixed}, all too-late conflicts involving $u$ are fixed, and by Lemma \ref{lem:conflictsAdded}, all conflicts created by the lower calls are added to the queue. 
    Similarly, for a vertex $(u,-) \in A$ participating in the raising stage, by Lemma \ref{lem:conflictFixed}, all too-early conflicts involving $u$ are fixed, and by Lemma \ref{lem:conflictsAdded}, all conflicts created by the raise call are added to the queue. 
\end{proof}

\myparagraph{Proving Correctness} Knowing that our queue maintains the conflicts correctly, we can quickly conclude the correctness of our algorithm. 

\begin{lemma}
    Given a strategy that selects at least one vertex every round, our algorithm terminates.
\end{lemma}

\begin{proof}
    Consider the lexicographic ordering of all possible queue contents (writing out the contents from early to late vertex). Note that processing a vertex $v$ increases the lexicographic order (of the entire contents of the queue). Because processing a vertex removes it from the queue, our lexicographic order gets later over time, and so we eventually terminate. 

\end{proof}

\begin{lemma}
    Our algorithm is correct (finishes with an \rlfmis{} on the new graph), for any policy for choosing $S$ that chooses at least one vertex each while loop iteration.
\end{lemma}
\begin{proof}
Because we terminate, we end with an empty queue. Therefore, no vertex has a conflict, and so we have the \rlfmis{}. 
\end{proof}

\section{Efficiency \label{sec:efficiency}}
\begin{table}
\centering
\begin{tabular}{c|c}
Symbol & Meaning \\
\hline
     $Y_j$ & $\frac{n \log n}{j}$ \\
    $Z_j$ & $\min(\Delta^2,(\frac{n \log n}{j})^2)$ \\
    $G_i$ & remaining graph immediately before \rlfmis{} step $i$ \\
    $\Delta$ & maximum degree in the graph (union of old and new graph) \\
    $b$ & size of update batch \\
    \end{tabular}
    \caption{Notation used in Section \ref{sec:efficiency}.}
\end{table}

In this section, we will prove that our algorithm is efficient (has low work and depth). 

First, we describe how we choose the subset $A$ (the subset of the queue we process each iteration). We will separate the vertices into shells by permutation value, and handle the shells sequentially (but in parallel within a shell). We will choose our shells such that a vertex is unlikely to have a neighbor within a shell, and so in expectation a vertex is only processed a constant number of times (compared to $O(\log n)$ times). 

Then, we will show that the cost of our algorithm is bounded in expectation by $O(\sum_{v \in S} Z_v)$, where $Z_v = \min(\Delta^2,(\frac{n \log n}{v})^2)$.\footnote{Note that $\Delta$ is the maximum vertex degree in the union of the old and new graphs.}
We conclude by using a modified version of the batch influence set analysis to directly bound $O(\sum_{v \in S} Z_v)$ by $O(\log^2 n \log \Delta)$.

\subsection{Shells}

Suppose that we handle the algorithm in shells (by permutation value), where the $i^{th}$ shell has size $s_i=(1+\frac{1}{\log^2 n})^{i-1}$.
The first shell starts at permutation time $1$, and for $i \ge 2$ note that the $i^{th}$ shell starts at permutation rank $p_i = 1 + \sum_{j=1}^{i-1} s_j = 1+\log^2 n(s_{i}-1)$. These shells are mainly for work analysis, but do affect the depth (for the worse). For notational convenience, let $Y_j = \frac{n \log n}{j}$ and $Z_j = \min(\Delta^2,(\frac{n \log n}{j})^2)$.

\begin{lemma}
    Observe that there are $O(\log^3 n)$ shells, and that the largest shell has size $O(\frac{n}{\log^2 n})$. \label{lem:largestShell}
\end{lemma}

\begin{proof}
    
Let $k$ be the number of shells, i.e., the largest $i$ such that $p_i \le n$.

\myparagraph{Largest shell size}
From $p_i \le n$, we have $1+\log^2 n(s_{i}-1) \le n$, so $s_{i} \le \frac{n-1}{\log^2 n}+1 = O(\frac{n}{\log^2 n})$.

\myparagraph{Number of shells}
Since $p_i = ((1+\frac{1}{\log^2 n})^{i-1} - 1) \log^2 n + 1$, the condition $p_k \le n$ requires $(1+\frac{1}{\log^2 n})^{k-1} \le \frac{n-1}{\log^2 n}+1 \le \frac{2n}{\log^2 n}$. Taking logarithms,
\[
  (k-1) \cdot \ln\!\left(1+\tfrac{1}{\log^2 n}\right) \le \ln\!\left(\tfrac{2n}{\log^2 n}\right) \le \ln n.
\]
Using $\ln(1+x) \ge x/2$ for $x \in (0,1)$, we have $\ln(1+\frac{1}{\log^2 n}) \ge \frac{1}{2\log^2 n}$, so
\[
  (k-1) \le \frac{\ln n}{1/(2\log^2 n)} = 2\log^2 n \cdot \ln n = O(\log^3 n)
\]
\[
    k=O(\log^3 n) \qedhere
\]

\end{proof}

\begin{lemma}
Let $v \in V$ with $p_i \le n/2$. The probability that $v$ has any (in or out) neighbor inside its shell in $G_{p_i}$ is at most $O(\frac{1}{\log n})$. \label{lem:lonelyShellLow}
\end{lemma}
\begin{proof}
Let $v$ be in shell $i$, which starts at permutation rank $p_i$ and has size $s_i$. We know that the number of surviving neighbors of $v$ at rank $p_i$ is at most $d \le c \frac{n \log n}{p_i}$ whp for some constant $c$ by Lemma \ref{lem:degreeBound}. Each surviving neighbor's rank is uniform over the $n - p_i$ remaining positions, and the shell (excluding $v$) occupies at most $s_i - 1$ of these. By a union bound,
\[
    \Pr[\text{$v$ has a neighbor in its shell}] \le d \cdot \frac{s_i}{n - p_i} \le \frac{cn\log n}{p_i} \cdot \frac{s_i}{n - p_i}.
\]
Substituting $p_i = 1+\log^2 n (s_{i}-1)$ and simplifying,
\[
    \frac{cn\log n}{1+\log^2 n (s_{i}-1)} \cdot \frac{s_i}{n - p_i}
    \le \frac{2cn\log n}{s_{i} \log^2 n} \cdot \frac{s_i}{n - p_i}
    = \frac{2cn}{(n - p_i) \log n},
\]
For $p_i \le n/2$, we have $n - p_i \ge n/2$, giving a bound of $\frac{4c}{\log n} = O(\frac{1}{\log n})$.
 \end{proof}

 \begin{lemma} Let $v \in V$ where $\pi(v) \ge n/2$. Then the probability that $v$ has any (in or out) neighbor in its shell in $G_{\pi(v)}$ is at most $O(\frac{1}{\log n})$.  \label{lem:lonelyShellHigh}
 \end{lemma}
\begin{proof}
Note that $v$ has degree $O(\log n)$ whp in $G_{n/2}$, and that $G_{\pi(v)} \subseteq G_{n/2}$. By Lemma \ref{lem:largestShell}, each shell contains at most $O(\frac{n}{\log^2 n})$ vertices. Since the neighbors of $v$ in $G_{n/2}$ are distributed uniformly at random among $\pi[n/2,n]$, we have that each neighbor has a $O(\frac{1}{\log^2 n})$ chance of appearing in $v$'s shell, so by a union bound, there is a $O(\frac{1}{\log n})$ chance of any neighbor appearing in $v$'s shell.
\end{proof}

\begin{lemma}
    Given a set of ready vertices $Q$, subselecting the elements $X$ in a particular shell can be done in $O(|X|)$ work and $O(\log n)$ depth.
\end{lemma}

\begin{proof}
We implement our set of ready vertices $Q$ as an array of $O(\log^3 n)$ bags, one per shell. Because shells are by permutation value, we can easily determine which bag a vertex belongs in. To extract a particular shell, we output all of the elements in that bag.
\end{proof}

Note that to store which shells are nonempty, we store $O(\log^2 n)$ words, each with $O(\log n)$ bits. We can thus iterate through $T$ nonempty shells in $O(T + \log^2 n)$ work and depth using standard table lookup.

\subsection{Depth}

Suppose we are processing shell $i$. Note that all earlier shells are finalized (their marks do not change), and vertices in shells later than $i$ do not affect the behavior of the queue when processing shell $i$. We will process all vertices in $Q$ in shell $i$ every round. Consider the set of vertices in the shell $V_i$, and the set of marked vertices in shell $i$ at the beginning of round $r$, $M_r$. Note that a while loop contains two rounds, one round of lowering followed by one round of raising. 

In this subsection, we use $M$ to denote whether or not a vertex is marked, rather than $l(u)=u$ as in Section \ref{sec:algorithm}. For the purposes of this subsection, we count raising and lowering as separate rounds (i.e. two rounds per while loop iteration). We consider round 1 to be the lowering stage of the first while loop iteration: marks and unmarks directly from the batch endpoints occur before the first round.

\begin{definition} We say that a vertex $v$ {\bf finalizes} at round $i$ if for all $x \ge i$, it holds that $v \in M_i \iff v \in M_{x}$, $v \notin M_i \iff v \notin M_x$ ($v$ does not change during round $i$ or any future round), and if $v \in M_{i-1} \Delta M_i$ ($v$ changed during round $i-1$). This means that the vertex stops switching between in and out of the marked set. Similarly, we call a set finalized if all vertices in the set are finalized. We say that a vertex finalizes in round 1 if it never changes mark during the while loop.
\end{definition} 

\begin{lemma}
    Suppose that $v$ ends marked and finalizes in round $r \ge 4$. Then there exists earlier neighbor $u$ that finalizes in round $r-1$.  \label{lem:u_exists}
\end{lemma}
\begin{proof}
Because $v$ finalizes in round $r$, and a vertex can only finalize as marked at the beginning of a lowering round, it follows that $r-1$ and $r-3$ are raising rounds, and that $r$ and $r-2$ are lowering rounds.

We note that it is either the case that $v$ was marked at the beginning of round $r-2$, but unmarked during round $r-2$ (Case 1), or that $v$ was unmarked at the beginning of round $r-2$ (Case 2), because if neither was the case, then $v$ would have finalized in round $r-2$ or earlier. 

\begin{enumerate}
        \item[Case 1:] $v$ was marked at beginning of round $r-2$, but unmarked during round $r-2$. Then at the beginning of round $r-2$, $v$ had an earlier marked neighbor $u$. Note that $u$ will end up unmarked (independence property), but finalizes no earlier than round $r-1$ because $u$ is marked at the beginning of $r-2$. Also note that $u$ finalizes no later than round $r-1$ because if $u$ was marked at the beginning of round $r-1$, then $u$ would prevent $v$ from marking, pushing back $v$'s finalization time. Therefore, $u$'s finalization time is $r-1$. 
        \item[Case 2:] At the beginning of round $r-2$, $v$ was not marked. Therefore, at the beginning of round $r-3$, $v$ has an earlier marked neighbor $u$. Note that $u$ settles by the beginning of round $r-1$ so that $v$ may finalize, but no earlier than round $r-1$ (as $u$ can only finalize at the beginning of a raising round), therefore $u$ finalizes in round $r-1$.  
    \end{enumerate}
\end{proof}

\begin{lemma} Suppose that $v$ finalizes at the beginning of round $r$. Then there exists an AE-path of length at least $r-4$ (in the new graph) that ends with $v$. \label{lem:exists_AE} 
\end{lemma}

\begin{proof}
    We prove by induction on the round that $v$ finalizes.

    Base Case: $r \le 4$: A single vertex is a path of length 0.
    
    Inductive step $(r > 4)$: we case on whether $v$ ends marked. 
    \begin{enumerate}
    \item[Case 1:] $v$ ends marked. By Lemma \ref{lem:u_exists}, there exists an earlier neighbor $u$ with finalization time $r-1$. 
    By the IH, there is an AE-path to $u$ of length $r-5$, so attaching $(u,v)$ yields an AE-path of length $r-4$ as desired. 

      \item[Case 2:] $v$ ends unmarked. Note that $r$ is a raising round. Let $w$ be the (final) eliminator of $v$. If $w$ finalized (at the beginning of) round $r-3$ or earlier, then $v$ would have finalized in round $r-2$ or earlier, a contradiction. Therefore $w$ finalized in round $r-1$ or later. Note that $w$ finalized no later than round $r-1$, for $v$ to finalize in round $r$. Applying the IH yields an AE-path of length $r-5$, which attached to $(w,v)$ gives an AE-path of length $r-4$.
    \end{enumerate}
    
\end{proof}

\begin{lemma}
    A single shell takes $O(\log n)$ rounds of the while loop to finish whp.  \label{lem:shell_log_rounds}
\end{lemma}

\begin{proof}[Proof of Lemma \ref{lem:shell_log_rounds}]
Since the length of all AE-paths is $c \log n$ whp for some constant $c$, whp there does not exist any AE-path of length $c \log n + 1$ or greater, so by Lemma \ref{lem:exists_AE} no vertex finalizes in rounds $c \log n + O(1)$ or later. Therefore, all vertices finalize within the first $O(\log n)$ rounds whp. 

Within a shell, note that all vertices in $Q$ are processed every round. For a vertex to add to the queue, by construction of our propagation algorithm, it must have changed mark. But a finalized vertex cannot change mark. Therefore, after $O(\log n)$ rounds, no more vertices are added to the queue, so the shell finishes after $O(\log n)$ rounds. \end{proof} 

\begin{lemma}
The total depth of the algorithm is $O(\log^5 n)$ whp. 
\end{lemma}
\begin{proof}
    We have $O(\log^3 n)$ shells, $O(\log n)$ rounds per shell, and $O(\log n)$ \depth{} for parallel primitives per round, for a total of $O(\log^5 n)$. 
\end{proof}

Note that our \depth{} bounds are loose upper bounds, for example, some shells will be empty if the batch size is small. Therefore, for $b=1$, our work bound (that we will achieve) of $O(\log^2 n \log \Delta)$ is smaller than our depth bound of $O(\log^5 n)$. For clarity we do not attempt to parameterize our depth by $b$. 

\subsection{Work}

We now will use the shells, and additional techniques, to bound the work of our algorithm. 

\begin{lemma}
   We have the following bounds.
  \begin{enumerate}
      \item Calling lower $j$ to $i$ costs $O(\min(\Delta,\frac{n \log n}{i}))$ whp.
      \item Calling raise $j$ from $i$ costs $O(\min(\Delta,\frac{n \log n}{i}))$ whp.
      \item The number of vertices in $N^h(i)$ is bounded by $O(\min(\Delta,\frac{n \log n}{i}))$ whp.
      \item Processing a vertex $v$ in the queue costs $O(\min(\Delta^2,(\frac{n \log n}{v})^2))$ work whp. 

  \end{enumerate}
\end{lemma}

\begin{proof}
In each of these statements, we take the min with $\Delta$, because the cost is proportional to the number of neighbors accessed, which is always bounded by $\Delta$.

Note that we can fetch the in-neighbors after time $i$, in constant time per neighbor after time $i$ by outputting from each bag in the desired range. Even though the earliest bag will contain neighbors earlier than $i$, we can filter these out efficiently, because the number of filtered neighbors is within a constant factor of the number of total neighbors fetched, because the bags are of doubling elimination time range. 

\begin{enumerate}
    \item Note that all neighbors of $j$ with elimination time from $i$ to $j$ have permutation value at least $i$, since elimination time is no more than permutation time. Thus, $N^-(j,i)$ existed entirely in the graph snapshot right before time $i$, which by Lemma \ref{lem:degreeBound} has $O(\min(\Delta,\frac{n \log n}{i}))$ max degree whp. 
    
    \item  The set of neighbors of $j$ with elimination time at least $i$ all have permutation time at least $i$, and therefore existed in the graph right before time $i$. Therefore, $O(\min(\Delta,\frac{n \log n}{i}))$ is a bound on the number of such vertices.
    \item Note that all neighbors in $N^h(i)$ have later deletion and permutation time than $i$, and therefore existed right before time $i$, so $O(\min(\Delta,\frac{n \log n}{i}))$ bounds the number of these.
    \item When vertex $v$ is processed, up to $O(\frac{n \log n}{v})$ neighbors (the size of $N^h(v)$) are raised/lowered, and a raise/lower costs $O(\min(\Delta,\frac{n \log n}{v}))$, for total 2-hop cost of $O(\min(\Delta^2,(\frac{n \log n}{v})^2))$.
\end{enumerate}
\end{proof}

\begin{lemma} If a vertex $v$ in shell $i$ has no neighbors in shell $i$, then it is processed (added and removed from queue) at most once by our algorithm. Otherwise, $v$ is processed at most $O(\log n)$ times whp. \label{lem:processCounts} \end{lemma} 

\begin{proof}
First suppose that $v$ has no neighbors in the shell. Note that a vertex can only be added to the ready bag by a neighbor being processed. Furthermore, when we are processing a shell, although we may add vertices from later shells, we only will process vertices in the current shell. Therefore, no vertex will add $v$ to the queue while processing shell $i$, so $v$ can only be processed if it was already in the queue (in which case, will be processed at most once). 

Now suppose that $v$ does have neighbors in the shell. By Lemma \ref{lem:shell_log_rounds}, we can spend at most $O(\log n)$ rounds processing a single shell. At worst, we will process $v$ in each of these rounds. Therefore, $v$ gets processed at most $O(\log n)$ times.

\end{proof}

\begin{lemma}
    Let $S_i$ be the influenced vertices in shell $i$. Processing shell $i$ will take $O(\sum_{v \in S_i} Z_v)$ expected work. Thus, summing over shells, we have $O(\sum_{v \in S} Z_v)$ expected work within all of the shells. \label{lem:zbound}
\end{lemma}

\begin{proof}
By the cost bounding lemma on processes, each time we process a vertex $v$ we pay at most $O( Z_v)$. 
By Lemmas \ref{lem:lonelyShellLow} and \ref{lem:lonelyShellHigh}, a vertex has a $O(\frac{1}{\log n})$ chance of having any neighbors in its shell. By Lemma \ref{lem:processCounts}, a vertex with no neighbors will be processed once, and a vertex with neighbors will be processed at most $O(\log n)$ times. Thus by linearity of expectation we pay $O(Z_v)$ per vertex in expectation when the dependence depth is shallow. 

Consider the (very unlikely) event where the dependence depth is not bounded by $O(\log n)$. Because within a shell, we process every vertex (in the shell) in every round, the worst possible number of rounds for a shell is $2n$. Therefore, our overall bound is this case is polynomial, which is negligible against a whp bound. 
Therefore, in expectation, we pay $O(Z_v)$ per vertex.

\end{proof}

\begin{lemma}
    Suppose that $E_1^?$ is the set of batch update edges (of size $b$), and that $V_1^?$ is empty. For $i < n/4$, we have $E[V_i^?] \le \frac{4bi}{n}$.  \label{lem:call-growth}
\end{lemma}

We will prove this lemma by instrumenting \mbox{InfluenceMIS}: instead of adding $u$ to the influence set, we will add $Z_u$ to a cost counter. 
Before proving this lemma, we make two important remarks.

\begin{enumerate}
\item If we added $1$ to our counter instead of $Z_v$, then we would be counting the size of the influence set.

\item This proof will proceed similarly to that of Lemma \ref{lem:expected_change}. However, there is a critical difference. The point of Lemma \ref{lem:expected_change} was to bound the sum of the undecided vertices and edges, without caring about the distribution between these two sets. In this lemma on the other hand, we are trying to show that given no initially undecided vertices and some undecided edges, the growth of the undecided vertices (transfer from undecided edges to vertices) is slow. 
\end{enumerate}

\begin{proof}[Proof of Lemma \ref{lem:call-growth}]

Base Case, $i=1$ (before any iteration), true because $V_1^?$ empty.

Inductive step: Consider $E[V_{i+1}^?]$. By the IH we have that $E[V_i^?] \le \frac{4bi}{n}$. We will consider the effect of decided edges in growing the undecided vertex set, then consider the effect of undecided edges. 

Suppose that a vertex $v \in V_i^?$ has $d_v$ edges incident to alive vertices not in $V_i^?$ right before time $i$ via decided edges. Picking $v$ as the next vertex would add $d_v$ vertices to the undecided set (and remove itself), but picking a vertex $u$ where $u \not\in V_i^?$ and $(u,v) \in E \setminus E_i^?$ would remove $v$ from the undecided set. Therefore, the net effect along decided edges is 0.

Along undecided edges, at best, we have a $\frac{2}{n-i+1}$ chance of selecting the edge, which would add $1$ vertex to the undecided set. Note that there are at most $b$ undecided edges.

Therefore, the expected undecided vertex set size after time $i$ is upper bounded by $\frac{4bi}{n} + \frac{2b}{n-i+1} \le \frac{4b(i+1)}{n}$. \end{proof}

\begin{lemma}
 $E[\sum_{v \in S} Z_v] \le O(b \log^2 n \log \Delta)$. \label{lem:bounding}
\end{lemma}

\begin{proof}

Note that we have two ways of counting $E_{\pi}[\sum_{v \in S} Z_v]$. The first is to sum over each vertex that actually enters the influence set. The second is to sum over the batch influence MIS construction, adding the 2-hop neighborhood charge if a vertex is added during that iteration. The latter view helps significantly for bounding the work. 

We note that in this proof, when we write $j$, we mean iteration $j$, we do not mean the vertex at the permutation position $j$. This is because the vertex at permutation position $j$ is a random choice, which is critical to our bound.

Observe that 
\begin{equation}E[\sum_{v \in S} Z_v]  = 
\sum_{j=1}^n Pr[\text{$\pi(j)$ enters influence set}] Z_{j}
\end{equation}

We will consider separately iterations $n/4$ through $n$ and iterations 1 through $n/4$. For iterations $n/4$ through $n$, note that there is a $O(\log n)$ degree bound whp on each vertex, so $Z_v \le O(\log^2 n)$ for all $v \in \pi[n/4,n]$ whp. Therefore, $E[\sum_{v \in S | \pi(v) \ge n/4} Z_v] \le O(|S| \log^2 n)$, which by Lemma \ref{lem:influenceBound} is bounded by $O(b \log^2 n)$.

Now consider iterations 1 through $n/4$. 
Each vertex not present in the permutation has an equal chance of being $\pi(j)$. A vertex will enter the influence set (and pay $Z_j$) iff it is in the undecided vertex set $V_j^?$. 
Therefore, there is a $\frac{|V_j^?|}{n-j+1} \le \frac{2|V_j^?|}{n}$ chance that a vertex joins the influence set during iteration $j$. Thus we have that $Pr[\pi(j) \text{ enters influence set}] \le \sum_{\pi[1:j)} \frac{1}{{n \choose j-1}} \frac{2|V_j^?|}{n} = \frac{2 E[|V_j^?|]}{n}$.

Thus we have that 
$$\sum_{j=1}^{n/4} Pr[\text{$\pi(j)$ enters influence set}] Z_{j} \le \sum_{j=1}^{n/4} \frac{2 E[|V_j^?|]}{n} Z_j.$$

Now by Lemma \ref{lem:call-growth}, we have that $E[|V_j^?|] \le \frac{4bj}{n}$, so we have that $$\sum_{j=1}^{n/4} \frac{2 E[|V_j^?|]}{n} Z_j \le \sum_{j=1}^{n/4} \frac{8bj}{n^2} \min \left(\Delta^2,\left(\frac{n \log n}{j} \right)^2 \right).$$

Note that $\Delta^2 < (\frac{n \log n}{j})^2$ when $j < \frac{n \log n}{\Delta}$, and so we split the summation accordingly.

Simplifying, we have that
\begin{align*}
\sum_{j=1}^{\frac{n \log n}{\Delta}} \frac{8bj}{n^2} \min\!\left(\Delta^2,\left(\frac{n \log n}{j}\right)^2\right)
  &= \sum_{j=1}^{\frac{n \log n}{\Delta}} \frac{8bj}{n^2} \Delta^2 \\
  &\le \left(\frac{n \log n}{\Delta}\right)^2 \frac{8b\Delta^2}{n^2} \\
  &= O(b \log^2 n).
\end{align*}

Using the log approximation of a partial sum, observe that
\begin{align*}
\sum_{j=\frac{n \log n}{\Delta}}^{n/4} \frac{8bj}{n^2} \min\!\left(\Delta^2,\left(\frac{n \log n}{j}\right)^2\right)
  &= \sum_{j=\frac{n \log n}{\Delta}}^{n/4} \frac{8bj}{n^2} \left(\frac{n \log n}{j}\right)^2 \\
  &= 8b \log^2 n \sum_{j=\frac{n \log n}{\Delta}}^{n/4} \frac{1}{j} \\
  &\le O\!\left(b \log^2 n \left(\log (n/4) - \log \frac{n \log n}{\Delta}\right)\right) \\
  &= O(b \log^2 n \log \Delta).
\end{align*}

Therefore, the overall bound for the sum until $j=n/4$ is $O(b \log^2 n \log \Delta)$. Thus the overall cost is $O(b \log^2 n \log \Delta)$. 

\end{proof}

\begin{theorem}
The total work of our algorithm is expected $O(b \log^2 n \log \Delta)$ on a batch of size $b$. 
\end{theorem}

\begin{proof}
    By Lemma \ref{lem:zbound}, the total cost of our algorithm is $O(\sum_{v \in S} Z_v)$, which by Lemma \ref{lem:bounding} is bound by $O(b \log^2 n \log \Delta)$.
\end{proof}

Note that if we did not use shells, we would significantly improve our depth bound, but at the cost of potentially processing every influenced vertex in every round. 
This would lead to an $O(\log^2 n)$ depth bound (whp) and $O(b \log^3 n \log \Delta)$ (expected) work bound. Note that the depth would be the same asymptotically as the depth for parallel static \rlfmis{}.

\begin{corollary}
There exists an algorithm for dynamically maintaining an \rlfmis{} with expected $O(b \log^3 n \log \Delta)$ work on a batch of size $b$, and $O(\log^2 n)$ depth whp.

\end{corollary}

On the other hand, making our algorithm fully sequential (processing the vertices one at a time, in priority queue order) would not further improve our work bound. Thus, our parallel algorithm is work-efficient. 

%% file: 7_conclusion.tex
\section{Conclusion}

-In this work, we have given an algorithm for maintaining a maximal independent set that supports batch-dynamic edge updates in $O(b \log^2 n \log \Delta)$ expected work and polylog depth \whp{} on a batch of size $b$.  
Noting that the expected neighborhood size of the $i^{th}$ vertex is $O(n/i)$ in expectation (in addition to being $O(\frac{n \log n}{i})$ whp), it would be interesting to try to reduce the number of log factors in our work. We note that a recent Masters' thesis by Krekelberg attempted to develop a $O(\log^2 n)$ time sequential dynamic algorithm for MIS, but appears to contain unproven technical lemmas and is incomplete \cite{krekelberg23fully}. 

Despite the high-level similarity between dynamic maximal matching and MIS, there is a gap currently between the state of the art dynamic MIS and maximal matching algorithms: the current best sequential dynamic lexicographic-first maximal matching algorithm has $O(\log \log^2 n)$ runtime \cite{das2026history}. Past work has also discussed the gap between these algorithms \cite{assadi2018fully}. Future work could attempt to find better bounds for MIS, or show lower bounds.

There has also been work on dynamic MIS on uniformly sparse graphs that relies on maintaining a low out-degree orientation \cite{onak2020fully}. Noting that parallel batch-dynamic algorithms were recently developed for low out-degree orientation \cite{blelloch26faster}, it would be interesting to try to parallelize Onak et al.'s techniques to improve the logarithmic term on uniformly sparse graphs.

%% file: 8_appendix.tex
\section{Appendix}

\subsection{Batch processing operations \label{sec:atomics}}

In our algorithms, we often will apply a map and filter to a list of pairs, as shown below. Let $L$ be an array, where each element is a pair $(u,n)$, for some variable $u$ and value $n$. In the function \mbox{setMins}, we want to set a field of $u$, $l(u)$, to be the minimum of the given pairs. If $u$ is not present in a pair, then we want to leave $l(u)$ unchanged. A first attempt at such a function would look like the following.

\begin{lstlisting}[basicstyle=\small\sffamily, keywordstyle=\bfseries, mathescape=true, escapeinside={@}{@}, morekeywords={parfor,for,if,else,foreach,return,while}, columns=flexible]
  setMinsUnsafe$(L)$:
   parfor $(u,n) \in L$:
       $l(u) \gets \min(l(u),n)$

\end{lstlisting}

However, because concurrent writes have an arbitrary output, and this involves concurrent writes, we need an alternative. Note that we can do the same effect in a concurrency-safe manner by using a groupBy (which is implemented with a semisort), as follows.

\begin{lstlisting}[basicstyle=\small\sffamily, keywordstyle=\bfseries, mathescape=true, escapeinside={@}{@}, morekeywords={parfor,for,if,else,foreach,return,while}, columns=flexible]
  setMinsSafe$(L)$:
    $L' \gets \mbox{groupBy}(L)$
    parfor $(v,N) \in L'$:
      $l(v) \gets \min(N)$
\end{lstlisting}

This has the same expected work and (whp) span as above (in binary forking), but is concurrency-safe. However, in pseudocode this is more complicated, and confuses reading. Thus, in this work, we will write pseudocode in the first form, with the understanding that to avoid concurrency issues an actual implementation would look more like the second form. We note that prior work in the field does the same \cite{BB25}.

\subsection{Challenge: Transferring Runtime Arguments of Prior Work \label{sec:transferHard}}

First, note that the batch influence set is not the union of the individual edge update influence sets, and so the work of the batch algorithm is not bounded by the sum of the work of the sequential updates. 
However, not even the proof approaches of CZ and BDHSS for bounding the runtime transfer to the batch setting. 
We will discuss in detail the issue of bounding the 2-level neighborhood of influence set vertices. 

When a vertex early in the permutation is in the influence set, its surviving 2-level neighborhood size is large, and processing this vertex is more expensive than processing a vertex late in the permutation. This neighborhood size is partially bounded in CZ and BDHSS by the time of the later endpoint of the dynamic edge update. 
Consider a set of updates $(u_i,v_i)$, where $u_i < v_i$, and $v_1 = \text{argmin}_i v_i$. Although we could process $O(1)$ vertices at time $v_1$, processing all $O(b)$ vertices expected to be in the batch influence set at time $v_1$ would be too costly. It is technically difficult to assign responsibility for influenced vertices to the edge that directly influenced them (or influenced via some path) because some vertices are influenced by more than one edge and since the influence set itself is randomly selected by the same randomness that is giving the vertices their relative positions in the permutation.

\subsection{Verifying an LFMIS \label{sec:lfmisVerify}}

\begin{lemma}
    We have that a set $M$ is the LFMIS on $\pi$ iff the following two conditions hold:
    \begin{enumerate} 
    \item for all $u \not\in M$, there exists $v \in N(u) \cap M$ with $\pi(v) < \pi(u)$ (all unmarked vertices have an earlier marked neighbor). This encapsualtes maximality and the lexicographic first condition. Note that $N(u)$ means the neighbors of $u$. 
    \item $M$ is independent (for all $u, v \in M$, $(u,v) \not\in E$) 
    \end{enumerate}
    \label{lem:lfmis}
\end{lemma}

\begin{proof}
For the forward direction of the implication, suppose that $M$ is the LFMIS. Then all unmarked vertices have an earlier marked neighbor (an eliminator). Furthermore, no two adjacent vertices are marked. 

Now for the backward direction of the implication, suppose that for all $u \not\in M$, there exists $v \in N(u) \cap M$ with $\pi(v) < \pi(u)$, and that for all $u,v\in M$, that $(u,v) \not\in E$. We perform induction on the loop of the LFMIS construction.

Base case: the first vertex. Note that $u=\pi(1) \in \ $LFMIS. Suppose BWOC that $u \not\in M$. Then there exists an earlier marked neighbor of 1, but there are no earlier neighbors of 1 at all, resulting in a contradiction. Therefore $u \in M$.

Inductive step: suppose that the LFMIS and $M$ agree on vertices with $\pi$ values 1 through $r-1$, consider $v_r=\pi(r)$. 

\begin{enumerate}
    \item[Case 1:] $v_r$ is not in the LFMIS. Suppose BWOC that $v_r$ is in $M$. By construction of the LFMIS, note that there exists an earlier neighbor $v$ in the LFMIS. By the IH, $M$ and the LFMIS agree on $v$, so $v \in M$. But $v$ and $v_r$ being adjacent and both marked this contradicts the condition that $M$ is independent, thus $v_r \in M$. 
    \item[Case 2:] $v_r$ is in the LFMIS. Suppose BWOC that $v_r$ not in $M$. Then by our requirement we have an earlier neighbor $v \in M$. By the IH, $v$ would be in the LFMIS, but then $v$ would have eliminated $v_r$, which is a contradiction. Thus $v_r \in M$.
\end{enumerate}

Therefore the LFMIS and $M$ agree. \end{proof}